%% file: main.tex
\documentclass[twoside]{article}

\usepackage[accepted]{aistats2025}

\usepackage{acro}
\usepackage{amsmath}
\usepackage{amssymb}
\usepackage{amsthm}
\usepackage{bm}
\usepackage{caption}
\usepackage{color}
\usepackage{comment}
\usepackage{hyperref}
\usepackage{graphicx}
\usepackage{natbib}
\usepackage{subcaption}
\usepackage{url}
\usepackage{cleveref}
\usepackage{scalerel}
\usepackage{stackengine}

\usepackage{IEEEtrantools}

\newtheorem{theorem}{Theorem}

\DeclareMathOperator*{\argmin}{arg\,min}
\newcommand\iid{ \; \smash{\stackrel{\mathrm{iid}}{\sim}} \; }

\acsetup{format/first-long=\itshape}
\DeclareAcronym{iid}{short = IID, long  = independent and identically distributed}
\DeclareAcronym{kl}{short = KL, long  = Kullback--Leibler}
\DeclareAcronym{mmd}{short = MMD, long  = maximum mean discrepancy}
\DeclareAcronym{rkhs}{short = RKHS, long  = reproducing kernel Hilbert space}
\DeclareAcronym{mcmc}{short = MCMC, long  = Markov chain Monte Carlo}
\DeclareAcronym{ode}{short = ODE, long  = ordinary differential equation}
\DeclareAcronym{sde}{short = SDE, long  = stochastic differential equation}
\DeclareAcronym{lvm}{short = LVM, long  = Lotka--Volterra model}
\DeclareAcronym{pcuq}{short = PCUQ, long  = Prediction-Centric Uncertainty Quantification}
\DeclareAcronym{mala}{short = MALA, long  = Metropolis-adjusted Langevin algorithm}
\DeclareAcronym{erk}{short = ERK, long  = extracellular signal regulated kinase}
\DeclareAcronym{rkip}{short = RKIP, long  = RAF kinase inhibitor protein}
\DeclareAcronym{mek}{short = MEK, long  = mitogen-activated protein kinase kinase}
\DeclareAcronym{meki}{short = MEKi, long  = MEK inhibitor}

\usepackage{xcolor}

\begin{document}

\twocolumn[

\aistatstitle{Prediction-Centric Uncertainty Quantification via MMD}

\aistatsauthor{ Zheyang Shen \And Jeremias Knoblauch \And  Sam Power \And Chris. J. Oates }

\aistatsaddress{ Newcastle University \And  University College London \And University of Bristol \And Newcastle University } ]

\begin{abstract}
    Deterministic mathematical models, such as those specified via differential equations, are a powerful tool to communicate scientific insight.
    However, such models are necessarily simplified descriptions of the real world.
    Generalised Bayesian methodologies have been proposed for inference with misspecified models, but these are typically associated with vanishing parameter uncertainty as more data are observed.
    In the context of a misspecified deterministic mathematical model, this has the undesirable consequence that posterior predictions become deterministic and certain, while being incorrect.
   Taking this observation as a starting point, we propose \emph{Prediction-Centric Uncertainty Quantification}, where a mixture distribution based on the deterministic model confers improved uncertainty quantification in the predictive context.
   Computation of the mixing distribution is cast as a (regularised) gradient flow of the maximum mean discrepancy (MMD), enabling consistent numerical approximations to be obtained.
   Results are reported on both a toy model from population ecology and a real model of protein signalling in cell biology.
\end{abstract}

\input{introduction}

\input{motivation}

\input{methodology}
\input{results}
\input{discussion}

\input{acknowledgements}

\bibliographystyle{abbrvnat}
\bibliography{bibliography}

\newpage
\onecolumn

\onecolumn
\appendix

\input{appendix}

\end{document}

%% file: introduction.tex
\section{INTRODUCTION}

Deterministic modelling has long been a staple in scientific and engineering disciplines, providing a simplified yet powerful framework in which complex systems can be understood. 
Unfortunately, such models are typically misspecified and, by their very nature, do not attempt to capture more complex aspects of the real world.
This can lead to misleading inferences and predictions, and is a particular problem across disciplines as diverse as climatology, epidemiology, and cell biology, where deterministic and simulation-based models are routinely used \citep[see][]{hermans2022crisis}.

A na\"{i}ve and often reflexive remedy proposed for addressing such robustness concerns is (some form of) Bayesian inference \citep[e.g.,][]{beaumont2010approximate, csillery2010approximate}.
Here, the hope is that the influence of a well-specified prior, coupled with an appropriate measurement error model, can counteract the failure to model the often complex structure that remains in the residual.
Unfortunately, this approach is flawed: As more data are observed, the influence of the prior is diminished and the Bayesian posterior concentrates around a single best-fitting parameter value \citep[see e.g.,][]{ghosal2000convergence, van2000asymptotic, ghosal2007convergence}. 
This loss of epistemic uncertainty is a severe limitation, as being pushed towards placing all belief on a singular parameter value leads to arbitrarily confident deterministic predictions being generated under a model that is fundamentally misspecified.

Despite considerable research into robust and/or generalised Bayesian inference, it remains the case that virtually all existing methodologies suffer the same problem:
While they provide robustness in the face of misspecified observations, noise, and outliers, their posteriors still collapse towards a single best-fitting parameter value once sufficient data are observed
\citep[see][]{chernozhukov2003mcmc, miller2021asymptotic,matsubara2022robust,matsubara2023generalized, frazier2024impact}.
To bridge the gap between powerful deterministic modelling and probabilistic uncertainty quantification, a change in perspective may be required.

This paper presents a new approach to probabilistic inference and prediction in the context of misspecified deterministic mathematical models, capable of providing non-vanishing epistemic uncertainty as more data are observed. 
This is achieved by inducing uncertainty in the parameter space through the implied predictive model, which is now interpreted as a mixture component within a mixture model. 
A variational criterion is proposed, based on the \ac{mmd}, which provides robustness to data outliers and enables computation via (regularised) gradient flow of the \ac{mmd}.
Practically, the result is an interacting particle system whose limiting distribution quantifies parameter uncertainty in a manner that is relevant to the predictive context.
As we show in a range of experiments, this new methodology is particularly appealing when one fits noisy data to a deterministic and possibly misspecified mathematical model.

%% file: motivation.tex
\section{MOTIVATION}

Bayesian inference is often promoted as a silver bullet for epistemic uncertainty quantification. 
However, to produce well-calibrated uncertainty, Bayesian updating 
requires the likelihood function to be well-specified \citep{walker2013bayesian, bissiri2016general}.
This is rarely the case when working with a deterministic mathematical model, as properly acknowledging the complexities of real-world data almost always requires stochasticity to feature at some level in the model.
The situation is not easily resolved, as from a scientific modelling point of view it can be much more challenging to design an accurate stochastic model compared to using traditional established techniques to design an interpretable deterministic model.

A na\"{i}ve remedy would  be to replace parameter inference and prediction with a form of---possibly black box---density estimation \citep[e.g.][]{fong2023martingale}.
Unfortunately, this would not be a useful solution in many scientific enquiries, as the structure provided by even an approximate model can be essential in constraining predictions to be at least `in the right ballpark'.
Further, many mathematical models are explicitly \emph{causal} models; for example, a differential equation model of cell signalling may describe biological mechanisms governing the signalling dynamics.
A causal model enables prediction for the effect of interventions which change how the system operates, such as the use of molecular drug treatment to modulate cell signalling (an example is presented in \Cref{subsec: engineering}), which density estimation alone cannot.

A second angle of attack that may seem appealing at first glance is to appeal to generalised Bayesian methods \citep{knoblauch2022optimization}. 
These can produce posteriors that are robust to  outliers when the data-generating model is misspecified---a feature that is particularly pertinent to deterministic mathematical models.
Yet, these methods exhibit a limitation similar to Bayesian inference, in that parameter uncertainty vanishes as more data are observed. 
This is due to the fact that generalised Bayesian procedures quantify parameter uncertainty in relation to the fit of data to a general loss function---rather than the implied predictive distribution on observables.
Such behaviour might be desirable if the data-generating mechanism perfectly aligns with a specific parameter value of the posited model, but it significantly worsens the predictive capacity of deterministic models that are misspecified, as incorrect predictions can become arbitrarily confident when more and more data are collected. 
This could in principle have serious real-world implications; for example, in \Cref{subsec: engineering} we present a case study in which over-confidence in the efficacy of a drug treatment results from applying Bayesian and generalised Bayesian methods to a misspecified deterministic cell signalling model.

\section{PREDICTION-CENTRIC UNCERTAINTY}

This section describes our \emph{prediction-centric} perspective on inference and prediction with a (possibly misspecified) deterministic mathematical model.
As with Bayesian and generalised Bayesian procedures, the approach that we present next produces a distribution over parameter values as its output. 
However, unlike these existing approaches, the parameter uncertainty in our methodology will not be defined relative to the \textit{average data fit} of the model, but relative to its \textit{predictive fit} instead. 
Our inspiration is drawn from the optimisation-centric view on Bayesian inference, in line with many contemporary extensions of Bayesian reasoning \citep{knoblauch2022optimization}. 
To explain this perspective, we define the parameter space $\Theta$, and consider observations $y_i \in \mathcal{Y}$ for $i=1,2,\dots n$ with $y_{1:n} = (y_1, y_2, \dots y_n)$, eschewing all considerations of measurability to simplify presentation throughout. 
The unsuitability of a given parameter $\theta \in \Theta$ for the observations $y_{1:n} \in \mathcal{Y}^n$ is quantified by a \emph{loss function} $\mathsf{L}_n:\Theta \times \mathcal{Y}^n \to \mathbb{R}$.
Further, we denote by $\mathcal{P}(\Theta)$ the set of probability distributions on $\Theta$, by $Q_0 \in \mathcal{P}(\Theta)$ a prior probability distribution over $\Theta$, by $\operatorname{KL}(Q, P)$ the \ac{kl} divergence between two probability distributions $Q, P \in \mathcal{P}(\Theta)$, and by $\lambda_n > 0$ a scalar.
In \citet{knoblauch2022optimization}, it is argued that a principled recipe for obtaining (generalised) posteriors over  $\Theta$ is given by the (assuming uniqueness) minimiser 
\begin{IEEEeqnarray}{rCl}
    Q_n^{\dagger} & = &
    \argmin_{Q \in \mathcal{P}(\Theta)} \underbrace{\int \mathsf{L}_n(\theta, y_{1:n})\: \mathrm{d} Q(\theta)}_{\text{average data fit}} +  \underbrace{\lambda_n\operatorname{KL}(Q, Q_0)}_{\text{regularisation}}
    , \qquad
    \label{eq:ocgbi}
\end{IEEEeqnarray}
which recovers the standard Bayesian posterior distribution when $\mathsf{L}_n(\theta, y_{1:n}) = -\log p_{\theta}(y_{1:n})$ is the negative log-likelihood arising from a statistical model $P_{\theta} \in \mathcal{P}(\mathcal{Y})$, and $\lambda_n = 1$.
In the motivating context of deterministic mathematical models, the \textit{statistical} model $P_\theta$ might represent noise-corrupted observations of the deterministic model output.
For various other choices of $\mathsf{L}_n$ and $\lambda_n$, generalised Bayesian methods can produce robust posteriors suitable for dealing with certain forms of statistical model misspecification \citep[see][]{
hooker2014bayesian,
ghosh2016robust,
knoblauch2018doubly, 
schmongeneralized, 
cherief2020mmd,
dellaporta2022robust, husain2022adversarial,
altamirano2023robust,
altamirano2023robustGP,  
duranoutlier}.

For all reasonable choices of $\mathsf{L}_n$ and $\lambda_n$, the relative importance of the first term
in \eqref{eq:ocgbi} increases with $n$. As a result, $Q_n^{\dagger}$ increasingly allocates most of its probability mass around a single (assuming uniqueness) best-fitting  parameter $\theta^\dagger$ \citep{miller2021asymptotic}. 
This effect is particularly problematic if the loss is defined relative to a parameter $\theta$ that indexes a misspecified model $P_{\theta}$:  now, the implied posterior predictive $\int P_{\theta} \: \mathrm{d}Q_n^{\dagger}(\theta)$ will quickly collapse to the plug-in predictive $P_{\theta^\dagger}$, and will effectively predict from a singular element of the misspecified model class $\{ P_{\theta} : \theta \in \Theta\}$. 
As a result, posteriors $Q_n^{\dagger}$ computed via \eqref{eq:ocgbi} will exhibit poor predictive uncertainty.
For a formal proof and additional findings surrounding this phenomenon, we refer the reader to \citet{mclatchie2024predictive}.

To remedy the predictive drawbacks of generalised Bayesian methods, we develop an alternative approach that instead quantifies parameter uncertainty relative to \emph{predictive} capability, and replaces the role of the average data fit with a form of predictive fit. 
To achieve this, we first denote the \acf{mmd} between probability distributions $P,Q$ as $\operatorname{MMD}(P,Q)$ \citep[see e.g.][]{gretton2012kernel}, and denote the empirical measure of the dataset as $P_n = \frac{1}{n}\sum_{i=1}^n \delta_{y_i}$.
Then we compute the (assuming uniqueness) minimiser
\begin{align}
    Q_n  = \argmin_{Q \in \mathcal{P}(\Theta)} \; \frac{1}{2} \underbrace{\operatorname{MMD}^2( P_n , P_{Q} )}_{\text{predictive fit}}
    +
    \underbrace{\lambda_n \operatorname{KL}(Q, Q_0)}_{\text{regularisation}}
    \label{eq:prediction-centric-uq}
\end{align}
where 
\begin{align}
P_Q = \int P_{\theta} \: \mathrm{d}Q(\theta) \in \mathcal{P}(\mathcal{Y})  \label{eq: mix model}
\end{align}
denotes a mixture model, whose components are instances $P_\theta$ of the statistical model and whose mixing distribution is $Q$, and  $\lambda_n > 0$ a constant to be specified.
Throughout the remainder of this paper, procedures defined by minimising \eqref{eq:prediction-centric-uq} will be referred to as \ac{pcuq}.
Before discussing related work, we briefly comment on the main features of \ac{pcuq}:

\paragraph{Mixture Model}
The mixture model $P_Q$ in \eqref{eq: mix model} expands predictive potential beyond the original statistical model $P_\theta$, which may be useful if the original statistical model is misspecified.
On the other hand, $P_Q$ retains the original statistical model $P_\theta$ as the special case where $Q = \delta_\theta$, which in principle enables $P_\theta$ to be exploited when the model is well-specified.

\paragraph{Predictive Fit}
The use of \ac{mmd} in \eqref{eq:prediction-centric-uq}, as opposed to other statistical divergences, confers outlier-robustness to \ac{pcuq}, which may be valuable in the misspecified context (c.f. \Cref{subsec: robust}), and the use of the \emph{squared} \ac{mmd} carries computational advantages, enabling the use of powerful emerging sampling methods based on gradient flows \citep{wild2023rigorous} (c.f. \Cref{subsec: q_n}).

\paragraph{Regularisation}
Similarly to \eqref{eq:ocgbi}, the regularisation term involves a reference distribution $Q_0$ (i.e. the prior, in the standard Bayesian context).
The functional role of $Q_0$ in \ac{pcuq} is explored in \Cref{ap: theory}, where we derive the perhaps surprising result that $Q_0$ acts on $Q_n$ in essentially the same way that $Q_0$ acts on Gibbs measures like $Q_n^{\dagger}$, as a reference measure in a Radon--Nikodym derivative \citep{bissiri2016general, knoblauch2022optimization}(c.f. \Cref{lem: characterise Qn} in \Cref{ap: theory}).
That is, one can reason about `updating belief distributions' using \ac{pcuq}.

\section{RELATED WORK}

The ideas we pursue in this work are linked to, and in some ways improve upon, a collection of seemingly disparate previous contributions in both classical statistics and machine learning. 

First amongst these is \textit{nonparametric maximum likelihood}, which minimises $Q \mapsto -\frac{1}{n}\sum_{i=1}^n  \log p_Q(y_i)$ where $p_{Q}$ is a density for $P_Q$ \citep[see Chapter 5 of][]{lindsay1995mixture}. 
This objective approximates $\operatorname{KL}(P_\star, P_Q)$ when $y_{1:n}$ is a collection of $n$ independent samples from $P_\star \in \mathcal{P}(\mathcal{Y})$, and can therefore be interpreted as a version of \eqref{eq:prediction-centric-uq} where $\lambda_n = 0$ and $\operatorname{MMD}^2$ is replaced by $\operatorname{KL}$.
The lack of regularisation causes several issues, including computational difficulties and non-identifiability \citep[see e.g.][]{laird1978nonparametric}, as the minimising measure will generally be fully atomic, see \citet[][e.g. Theorem 21 in Chapter 5]{lindsay1995mixture} and \citet{jordan2015convex}.

One way of addressing these shortcomings is via regularisation.
While the link to nonparametric maximum likelihood is never explicit, regularisation like this has been enforced by adding a \ac{kl} divergence as in \eqref{eq:prediction-centric-uq}  and by constraining $Q$ to a parametric variational family $\mathcal{Q}$ \citep{jankowiak2020parametric, jankowiak2020deep, sheth2020pseudo}, resulting in objectives that are variants of 
\begin{align}
    \argmin_{Q \in \mathcal{Q}} 
    -\frac{1}{n} \sum_{i=1}^n \log \left(p_{Q}(y_i)^{\alpha}\right) + \lambda_n \operatorname{KL}(Q, Q_0).
    \label{eq:half-way-eq}
\end{align}
For example, the choice $\lambda_n = \frac{\alpha}{n}$ can be linked to approximation of the standard Bayesian posterior via $\alpha$-divergences \citet{li2017dropout, villacampa2020alpha}.
Similarly, this objective has been considered in the context of Gaussian processes and deep Gaussian processes with $\alpha = 1$ and various choices for $\lambda_n$ \citet{jankowiak2020deep,jankowiak2020parametric,sheth2020pseudo}.
More recently, \citet{crucinio2022solving} used a similar formulation to solve Fredholm integral equations.
Enriching this with a solid theoretical foundation, \citet{masegosa2020learning} and \citet{morningstar2022pacm} motivated targeting objectives like \eqref{eq:half-way-eq} via PAC-Bayesian bounds.
While they result in superior predictive performance relative to standard Bayesian and variational schemes, the applicability of objectives like \eqref{eq:half-way-eq} is limited; apart from  Gaussian-type likelihoods where the log integral $\log p_{Q}(y_i)$ has a closed form, computation can become impractical. 
In particular, approximating the log integral via samples from $Q$ can yield a highly biased approximation, and generally renders this methodology impractical.

This substantive drawback prompts us to seek inspiration from an ongoing line of research into generalised and post-Bayesian methods \eqref{eq:ocgbi} popularised by \citet{bissiri2016general} and \citet{knoblauch2022optimization}.
Specifically, we will go beyond loss functions that focus on goodness-of-fit, and instead assess predictive fit.
Further, we seek a predictive loss that is both computationally feasible and robust.
To achieve this, we follow  \citet{cherief2020mmd} and \citet{alquier2024universal}, and assess predictive fit in  \eqref{eq:prediction-centric-uq} via the (squared) \ac{mmd}.
Beyond tackling the computational challenges of the intractable log integral $\log p_{Q}(y_i)$, the MMD allows us to obtain inferences that are more robust to outliers and model misspecification \citep[see][]{briol2019statistical, cherief2022finite, alquier2023estimation}.

%% file: methodology.tex
\section{METHODOLOGY}

Throughout \Cref{subsec: pre via mmd,subsec: robust,subsec: q_n}, we will simplify the presentation and assume that $y_i$ are sampled independently from some unknown population distribution $P_\star \in \mathcal{P}(\mathcal{Y})$.
The case of non-independent data, where data $y_i$ are allowed to depend on covariates $x_i$, is deferred to \Cref{subsec: non IID}.

\subsection{Predictive Fit via MMD}
\label{subsec: pre via mmd}

The predictive fit term in \eqref{eq:prediction-centric-uq} is based on the \ac{mmd}, which will simultaneously confer computational efficiency and robustness to our method.
To define it, we let $k : \mathcal{Y}  \times \mathcal{Y}  \rightarrow \mathbb{R}$ be a kernel\footnote{A function $k : \mathcal{Y}  \times \mathcal{Y}  \rightarrow \mathbb{R}$ is  a \emph{kernel} if $k(u,v) = k(v,u)$ for all $u,v \in \mathcal{Y} $, and $\sum_{i=1}^m \sum_{j=1}^m w_i w_j k(u_i,u_j) \geq 0$ for all $w_1,\dots,w_m \in \mathbb{R}$, $u_1,\dots,u_m \in \mathcal{Y} $, and $m \in \mathbb{N}$.} %
and denote by $\mathcal{H}(k)$ the associated \ac{rkhs}  \citep[see][for background]{berlinet2011reproducing}.
For a given kernel $k$ and a probability distribution $P \in \mathcal{P}(\mathcal{Y})$, we can now define the \textit{kernel mean embedding} as\footnote{Throughout, we will assume this embedding exists as a strong (Bochner) integral. This always holds for the most popular choices of kernels (such as Gaussian and Mat\'ern kernels).}
\begin{IEEEeqnarray}{rCl}
\mu_k(P) & := & \int k(\cdot , y) \; \mathrm{d}P(y) \in \mathcal{H}(k).
\nonumber
\end{IEEEeqnarray}
The divergence of a candidate $P \in \mathcal{P}(\mathcal{Y})$ from the data-generating distribution $P_\star$ can be quantified using \ac{mmd}; a pseudometric defined as the \ac{rkhs}-norm  between the pair of kernel mean embeddings, i.e.
\begin{IEEEeqnarray}{rCl}
    \operatorname{MMD}(P_\star, P) & = &
    \| \mu_k(P_\star) - \mu_k(P) \|_{\mathcal{H}(k)} \:.
    \nonumber
\end{IEEEeqnarray}
The \ac{mmd} is a proper metric if $k$ is a \emph{characteristic} kernel \citep{sriperumbudur2011universality}; our use of \ac{mmd} is justified by its interpretation as a statistical divergence induced by a \emph{proper scoring rule} \citep{dawid1986probability}.
The mixture model $P_Q$ in \eqref{eq:prediction-centric-uq} has kernel mean embedding
\begin{IEEEeqnarray}{rCl}
    \mu_k(P_Q) & = & \iint k(\cdot,y) \; \mathrm{d}P_\theta(y) \mathrm{d}Q(\theta) = \int \mu_k(P_\theta) \; \mathrm{d}Q(\theta), 
    \nonumber
\end{IEEEeqnarray}
so the \ac{mmd} between $P_\star$ and $P_Q$ can be written as
\begin{IEEEeqnarray}{rCl}
    \operatorname{MMD}^2(P_\star, P_Q)
    & = &
    \left\| \int \left\{\mu_k(P_\star) - \mu_k(P_\theta) \right\}\mathrm{d}Q(\theta) \right\|_{\mathcal{H}(k)}^2 %
    \nonumber \\
    & = &
    \iint \kappa_{P_\star}(\theta,\vartheta) \; \mathrm{d}Q(\theta) \mathrm{d}Q(\vartheta), \label{eq: quadratic form}
\end{IEEEeqnarray}
where the last step follows from expanding the norm,   using the inner product, and exchanging inner product and integral. 
Here, the resulting
$\kappa_{P_\star} : \Theta \times \Theta \rightarrow \mathbb{R}$ is  a kernel on $\Theta$, and given by 
\begin{IEEEeqnarray}{rCl}
    \kappa_{P_\star}(\theta,\vartheta)
    & = &
    \langle \mu_k(P_\star) - \mu_k(P_\theta) , \mu_k(P_\star) - \mu_k(P_\vartheta) \rangle_{\mathcal{H}(k)} \:.
    \nonumber
\end{IEEEeqnarray}
This reveals one possible interpretation of \eqref{eq: quadratic form} as a \textit{kernel Stein discrepancy} \citep{chwialkowski2016kernel,liu2016kernelized,gorham2017measuring} corresponding to the \textit{Stein kernel} $ \kappa_{P_\star}$ \citep{oates2017control}.
One of the implications is that, \textit{if}  $P_\star = P_{\theta_\star}$ for some unique $\theta_\star \in \Theta$, then \eqref{eq: quadratic form} is uniquely minimised by $Q = \delta_{\theta_\star}$ provided $k$ is a characteristic kernel.

\subsection{Estimation, Regularisation and Robustness}
\label{subsec: robust}

Of course, the true data-generating distribution $P_\star$ in \eqref{eq: quadratic form} is unknown and must be approximated. 
Following \eqref{eq:prediction-centric-uq}, we use the empirical distribution $P_n$ in lieu of $P_\star$.
For this special case, we find
\begin{align}
\kappa_{P_n}(\theta,\vartheta) & \stackrel{+C}{=} \iint k(y,y') \; \mathrm{d}P_\theta(y) \mathrm{d} P_\vartheta(y')
 \label{eq: explicit mmd} \\
& \hspace{-30pt} -\frac{1}{n} \sum_{i=1}^n \int k(y_i,y)  \mathrm{d}P_\theta(y)
-\frac{1}{n} \sum_{i=1}^n \int k(y_i,y)  \mathrm{d}P_\vartheta(y) ,
\nonumber
\end{align}
where \smash{$\stackrel{+C}{=}$} denotes a $\theta$- and $\vartheta$-independent additive constant.
Specific choices of $k$ and $P_{\theta}$ lead to tractable integrals in \eqref{eq: explicit mmd}, while computational strategies are available when numerical approximation is required; see \Cref{app: scores fo iid}.
One can contrast the bias introduced in the Monte Carlo approximation of log integrals $\log p_{Q}(y_i)$ that occurs in approximating objectives of the family in \eqref{eq:half-way-eq} with straight-forward unbiased approximation of the squared \ac{mmd} via Monte Carlo.
Further, Huber robustness of the \ac{mmd} is immediate from \eqref{eq: explicit mmd}, since outlier data $y_i$ far from the effective support of $P_\theta$ contribute negligibly to \eqref{eq: explicit mmd} for typical choice of kernel $k$ \citep{alquier2024universal}.

A plug-in approximation necessitates additional regularisation, since otherwise minimisation of $Q \mapsto \operatorname{MMD}(P_n, P_Q)$ would result in a discrete distribution where, similarly to nonparametric maximum likelihood, each atom in the support would correspond to a value of $\theta$ that explains one of the data points well. 
One na\"{i}ve possibility for the regulariser is to use another (squared) \ac{mmd} between $Q$ and a reference distribution $Q_0 \in \mathcal{P}(\Theta)$.
This choice is appealing because it is again a quadratic form, so efficient convexity-exploiting algorithms can be used.
However, the (squared) MMD regulariser imposes a weak topology that is prone to the same issue of returning a distribution with finite support, and does not impose sufficient convexity for the gradient flow algorithms described in \Cref{subsec: q_n} to converge \citep[this was proven in Theorem 6 of][]{wild2023rigorous}.
Instead, we choose to employ \ac{kl} regularisation in \eqref{eq:prediction-centric-uq}, which ensures that $Q_n$ is absolutely continuous with respect to $Q_0$ (c.f. \Cref{lem: characterise Qn} in \Cref{ap: theory}).

The main challenge with using \ac{kl} regularisation in this context is that, until very recently, there were not efficient computational algorithms for solving regularised variational problems such as \eqref{eq:prediction-centric-uq}.
Fortunately, recent work on gradient flows provides a path forward \citet{wild2023rigorous}, which we explain next.

\subsection{Approximating $Q_n$ via Gradient Flow}
\label{subsec: q_n}

From \eqref{eq:prediction-centric-uq}, and now with $\Theta = \mathbb{R}^p$, the output $Q_n$ of the proposed \ac{pcuq} method is a minimiser of the entropy-regularised objective
\begin{align}
    \mathcal{F}_n(Q) &= \mathcal{E}_n(Q) + \lambda_n \int \log q(\theta) \; \mathrm{d}Q(\theta) , \label{eq: free energy}  
\end{align}
where the \emph{free energy} $\mathcal{E}_n(Q)$ is identical, after algebraic manipulation, to
\begin{align*}
    \mathcal{E}_n(Q) & \stackrel{+C}{=} \int v(\theta) \; \mathrm{d}Q(\theta) + \frac{1}{2} \iint \kappa_{P_n}(\theta, \vartheta) \; \mathrm{d}Q(\theta)\mathrm{d}Q(\vartheta) , \nonumber
\end{align*}
where $v(\theta) = - \lambda_n\log q_0(\theta)$, %
and where $q$ and $q_0$ are respectively densities for $Q$ and $Q_0$.
Sufficient conditions for existence and uniqueness of $Q_n$, for any choice of $\lambda_n > 0$, are presented in \Cref{thm: exist} of \Cref{ap: theory}. 
Based on this perspective, we can exploit recent advances in mean-field Langevin dynamics to numerically approximate $Q_n$, as we explain next.

For the entropy-regularised functional $\mathcal{F}_n$ \eqref{eq: free energy}, we can simulate a Wasserstein gradient flow via a McKean--Vlasov process \citep{ambrosio2008gradient}
\begin{align}
    \mathrm{d}\theta_t & = - \nabla_{\mathrm{W}} \mathcal{E}_n(Q^t)(\theta_t) + \sqrt{2\lambda_n}\mathrm{d}W_t, \label{eq: wgd_mv}\\
    \nabla_{\mathrm{W}} \mathcal{E}_n(Q^t)(\theta_t) & = \nabla v(\theta_t) + \int \nabla_1 \kappa_{P_n}(\theta_t, \vartheta) \; \mathrm{d}Q^{t}(\vartheta) \nonumber
\end{align}
where $Q^{t} = \mathsf{law}(\theta_t)$, $\nabla_{\mathrm{W}}$ denotes the Wasserstein gradient, $(W_t)_{t \geq 0}$ is a Wiener process on $\mathbb{R}^p$ and, for the bivariate function $\kappa_{P_n}$, the notation $\nabla_1 \kappa_{P_n}$ denotes differentiation with respect to the first argument. 
As \eqref{eq: wgd_mv} can be seen as the mean-field limit of an interacting particle system, we can discretise $Q^{t}$ into a system of $N$ evolving particles $\theta_t^1, \theta_t^2, \hdots, \theta_t^N$, whose evolution is governed by the following system of \acp{sde}:
\begin{align}
    \mathrm{d}\theta_t^i =& -\bigg( \nabla v(\theta_t^i) + \frac{1}{N-1} \sum_{j\ne i} \nabla_1 \kappa_{P_n}(\theta_t^i, \theta_t^j) \bigg) \mathrm{d}t \notag \\
    & \qquad + \sqrt{2\lambda_n}\mathrm{d}W_t^i, \label{eq: wgd_mv_ip}
\end{align}
where $(W_t^i)_{t \geq 0}$ are $N$ independent Wiener processes on $\mathbb{R}^p$. 
An Euler--Maruyama discretisation of \eqref{eq: wgd_mv_ip} incurs per-iteration computational complexity $O(n N^2)$ and storage complexity (with caching) of $O(n + N)$.
The only remaining technical challenge in simulating \eqref{eq: wgd_mv_ip} is calculation of the gradient $\nabla_1 \kappa_{P_n}$; this mirrors the challenge of evaluating the integrals defining $\kappa_{P_n}$ itself in \eqref{eq: explicit mmd}, and the same strategies discussed in \Cref{app: scores fo iid} can be applied.

Gradient flows are preferred in this paper for their ease of implementation, but other computational methods could be considered; see \Cref{ap: mcmc}.

\subsection{Extension to Dependent Data}
\label{subsec: non IID}

Consider now the (often more practically-relevant) setting of non-independent data, where each datum $y_i$ is associated with a covariate $x_i \in \mathcal{X}$ and generated according to an (unknown) conditional distribution $P_\star( \cdot | x_i)$.
Our task now involves a conditional model $\{P_\theta(\cdot | x)\}_{\theta \in \Theta}$ for each $x \in \mathcal{X}$.
To extend our methodology to this setting, we suppose that
\begin{align*}
    \{ (x_i,y_i) \}_{i=1}^n \; \iid \;  \bar{P}_\star(\mathrm{d}x, \mathrm{d}y) := \frac{1}{n} \sum_{i=1}^n \delta_{x_i}(\mathrm{d}x) \; P_0(\mathrm{d}y | x_i) 
\end{align*}
and consider the extended model 
\begin{align}
\bar{P}_\theta(\mathrm{d}x , \mathrm{d}y) := \frac{1}{n} \sum_{i=1}^n \delta_{x_i}(\mathrm{d}x) \; P_\theta(\mathrm{d}y | x_i)  \label{eq: joint disn}
\end{align}
so that $\bar{P}_\star, \bar{P}_\theta \in \mathcal{P}(\mathcal{X} \times \mathcal{Y})$.
Our approach then proceeds as before, but with $k$ now a kernel on the extended space $\mathcal{X} \times \mathcal{Y}$.
For example, if $\mathcal{X} \subset \mathbb{R}^{d_{\mathcal{X}}}$ and $\mathcal{Y} \subset \mathbb{R}^{d_{\mathcal{Y}}}$, we may consider the Gaussian kernel
\begin{align}
\hspace{-8pt} \resizebox{0.91\hsize}{!}{ $ \displaystyle k((x,y),(x' \hspace{-2pt} , y' )) =  \exp\left( - \frac{\|x-x'\|^2}{\ell_{\mathcal{X}}^2} - \frac{\|y-y'\|^2}{\ell_{\mathcal{Y}}^2} \right) $ }   \label{eq: Gauss kernel}
\end{align}
with bandwidths $\ell_{\mathcal{X}}$ and $\ell_{\mathcal{Y}}$ to be specified.
The lifting of a conditional density to a joint density via a \emph{plug-in} empirical distribution for the covariate was studied in \citet{alquier2024universal}.
The suitability of this approach hinges on the number $n$ of data being sufficiently large, and the characteristics of the dependence on covariates being appropriately modelled by the kernel.
The latter assumption can be removed by formally taking $\ell_{\mathcal{X}} \rightarrow 0$, as recommended in \citet{alquier2024universal}.
Strategies for numerically approximating $\nabla_1 \kappa_{P_n}$ extend to this non-independent setting; details are deferred to \Cref{app: gradients for non IID}.
A reduction in computational complexity also results from taking the $\ell_{\mathcal{X}} \rightarrow 0$ limit; see \Cref{app: ell x zero limit}.

%% file: results.tex
\section{EXPERIMENTS}
\label{sec: experiments}

For empirical assessment, we focus on the notable failure of standard and generalised Bayesian inference in the context of misspecified deterministic \ac{ode} models, asking whether the situation can be improved using \ac{pcuq}.
Examples of this abound; here we consider (a) the canonical Lotka--Volterra \ac{ode} model from population ecology, as a simple test-bed where the failure of Bayesian inference can be easily understood (\Cref{subsec: ecology}), and (b) a sophisticated system of \acp{ode} describing a cell signalling pathway, where prediction of cell response to a molecular treatment is required (\Cref{subsec: engineering}).
All experiments were performed on a 2023 MacBook Pro with 16 GB RAM. Code to reproduce these experiments can be downloaded from \url{https://github.com/zheyang-shen/prediction_centric_uq}.

\subsection{Illustration on a Lotka--Volterra Model}
\label{subsec: ecology}

As a simple test bed, we consider inference and prediction based on the deterministic \ac{lvm} in a context where data actually arise from a stochastically modified \ac{lvm}.

\paragraph{Lotka--Volterra Model}

As a prototypical model from population ecology, the \ac{lvm} describes the dynamical interaction between a prey ($u_1$) and predator ($u_2$) as a coupled system of differential equations 
\begin{align*}
    \frac{\mathrm{d}u_1}{\mathrm{d}x} & = \alpha u_1 - \beta u_1 u_2 , & u_1(0) = \xi_1, \\
    \frac{\mathrm{d}u_2}{\mathrm{d}x} & = \delta u_1 u_2 - \gamma u_2, & u_2(0) = \xi_2 ,
\end{align*}
for some $\alpha,\beta,\gamma,\delta, \xi_1, \xi_2 \geq 0$.
Real data arise as noisy observations of one or more species at discrete times; for our purposes we suppose that $y_{1:n}$ are noisy measurements of the populations $u = (u_1,u_2)$ at times $x_{1:n}$.
It is common to fit such models to data via the assumption of a Gaussian likelihood, which in this case would correspond to a Gaussian model with density
\begin{align}
    \hspace{-5pt} p_\theta(y_i | x_i) & = \prod_{i=1}^n \frac{1}{\sqrt{2 \pi \sigma^2}} \exp\left( - \frac{\|y_i - u_\theta(x_i)\|^2}{ 2 \sigma^2 } \right) , \label{eq: Gauss likelihood}
\end{align}
where parameters are $(\alpha,\beta,\gamma,\delta, \xi_1, \xi_2, \sigma)$, and we denote the dependence of the prey and predator populations on these parameters using $u_\theta$.
To improve visualisation, we consider inference only for $\theta_1 := \mathrm{logit}(\alpha)$ and $\theta_2 = \mathrm{logit}(\beta)$, with all other parameters fixed.
The distribution $Q_0$ was taken to be standard normal.

\paragraph{Failure of Bayesian Inference}

To explore the pitfalls of Bayesian inference when the \ac{lvm} is misspecified, we simulated predator-prey interactions using both the above \ac{ode} and a stochastic \ac{lvm} ($\epsilon_1, \epsilon_2 > 0$):
\begin{align*}
    \mathrm{d}u_1 &= (\alpha u_1 - \beta u_1 u_2) \; \mathrm{d}x + \epsilon_1 \; \mathrm{d}W_1, & u_1(0) = \xi_1 , \\
    \mathrm{d}u_2 &= (\delta u_1 u_2 - \gamma u_2) \; \mathrm{d}x + \epsilon_2 \; \mathrm{d} W_2, & u_2(0) = \xi_2 ,
\end{align*}
where stochasticity is used to represent the additional complexities of real predator-prey interactions that are not explicitly captured by the simple \ac{ode} model.
Observation noise was added to these simulations using the same Gaussian model in \eqref{eq: Gauss likelihood}, so that the observation model (at least) is correctly specified.
The values of all parameters are given in \Cref{app: extra experiments}.

\begin{figure*}[t]
\centering
    \includegraphics[width=\textwidth]{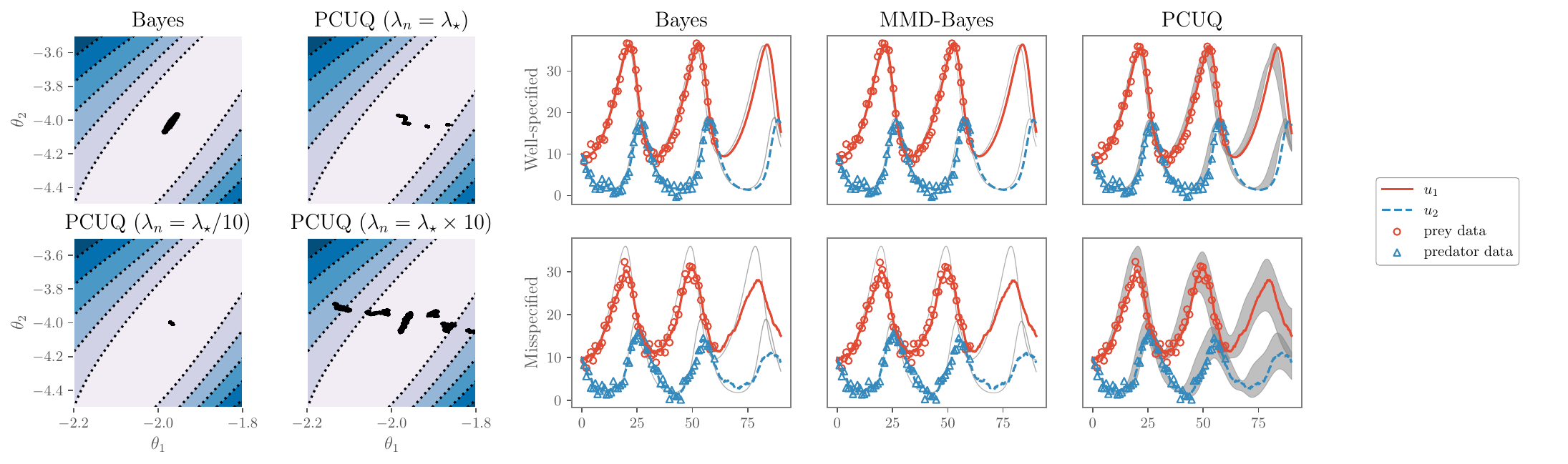}
    \caption{Lotka--Volterra model.  Left:  Contour plots depict the standard Bayesian posterior, superimposed with samples from the standard Bayesian posterior (Bayes) and the proposed method (PCUQ) with varying regularisation parameter $\lambda_n$.
    Right:  Predictive distributions in both the well-specified (top) and misspecified (bottom) context.
    (Lines depict true prey and predator populations, while the shaded regions depict the predictive quartiles for standard Bayesian inference, the MMD-Bayes method of \citet{cherief2020mmd}, and the proposed \ac{pcuq}.)
    }
    \label{fig: lv}
\end{figure*}

Bayesian inference is seen to work well when the model is well-specified, but to fail when predicting future predator and prey abundance when the model is misspecified.
This is because the Bayesian posterior concentrates around a single parameter configuration (\Cref{fig: lv}, top left panel), leading to concentration of the predictive distribution for $u(x)$ (\Cref{fig: lv}, middle column). 
Yet future predator and prey abundances do not coincide with these predicted values when the model is misspecified, rendering the Bayesian posterior grossly over-confident.
The same failure mode occurs for generalised Bayesian methods, which also concentrate around a single parameter. 
We illustrate this using the example of MMD-Bayes \citep{cherief2020mmd}, in the fourth column of \Cref{fig: lv} (for details, see \Cref{app: extra experiments}).

\paragraph{Prediction-Centric UQ}

Performing \ac{pcuq} in this context requires the extension to dependent data described in \Cref{subsec: non IID}.
The Gaussian kernel \eqref{eq: Gauss kernel} was employed with bandwidths $\ell_{\mathcal{X}} \rightarrow 0$ and $\ell_{\mathcal{Y}} = \sigma$ following the recommendation in \citet{alquier2024universal}; further discussion is contained in \Cref{app: ell x zero limit}.
The Gaussian kernel has the distinct advantage that all of the integrals appearing in \eqref{eq: explicit mmd} can be analytically evaluated, due to conjugacy with the Gaussian measurement error model \eqref{eq: Gauss likelihood}; see \Cref{app: gaussian calculations}.
The regulariser $\lambda_n$ was set to the value $\lambda_\star$ for which the spread of the \ac{pcuq} output $Q_n$ was similar to that of the standard Bayesian posterior $Q_n^\dagger$ when the model is well-specified (see \Cref{fig: lv}, left); a generally-applicable heuristic\footnote{A heuristic is necessary, as strategies for selecting $\lambda_n$ in generalised Bayesian methods remain a subject of ongoing research and debate, as discussed in detail in \citet{wu2023comparison, mclatchie2024predictive}.  Advances in this direction could potentially be deployed to \ac{pcuq}.} that requires only simulating data from the statistical model.
Since this system of \acp{ode} does not admit a closed-form solution, we solve the sensitivity equations to compute the gradients $\nabla_\theta \log p_\theta(y_i|x_i)$ as required to implement the gradient flow \eqref{eq: wgd_mv_ip}.
The predictive output from \ac{pcuq}, obtained with a particle system of size $N = 10$, is displayed in the right column of \Cref{fig: lv}.
This is qualitatively distinct from the standard Bayesian posterior, as would be expected.
Using \ac{pcuq}, the over-confidence of existing methodologies in the misspecified context is seen to be avoided.
Convergence diagnostics for the gradient flow are presented in \Cref{fig: lv3} and insensitivity of the predictions to increasing $N$ and decreasing $\lambda_\star$ is demonstrated in \Cref{fig: lv2,fig: lv4,fig: lv5}, all contained in \Cref{app: extra experiments}.

\subsection{Model Misspecification in Cell Signalling}
\label{subsec: engineering}

\begin{figure*}[t!]
    \centering

    \stackunder{ \includegraphics[width=0.3\linewidth]{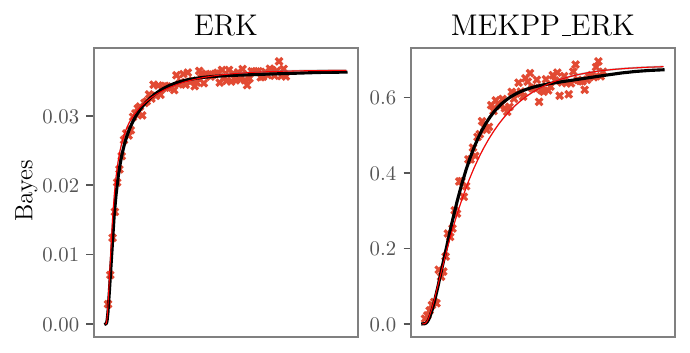} }{ \includegraphics[width=0.3\linewidth]{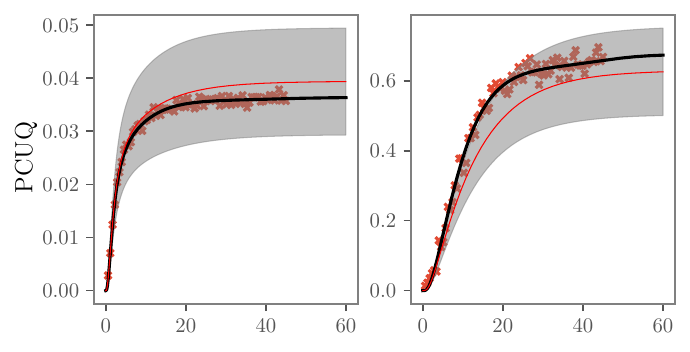} }%
    \hspace{10pt}%
    \raisebox{-0.45\height}{\includegraphics[width=0.27\linewidth,clip,trim = 21cm 5cm 17.5cm 5cm]{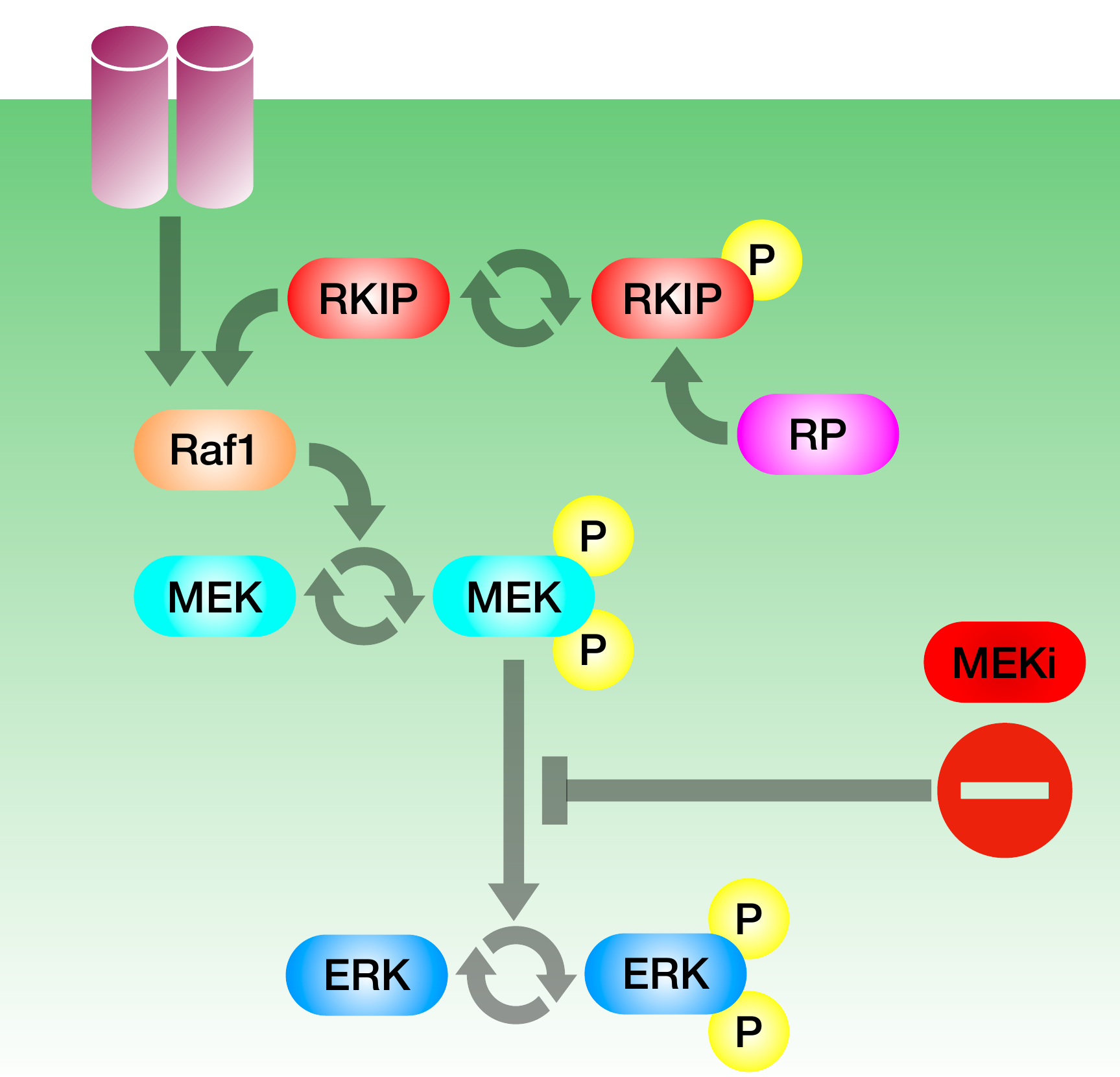}}%
    \hspace{10pt}%
    \stackunder{ \includegraphics[width=0.3\linewidth]{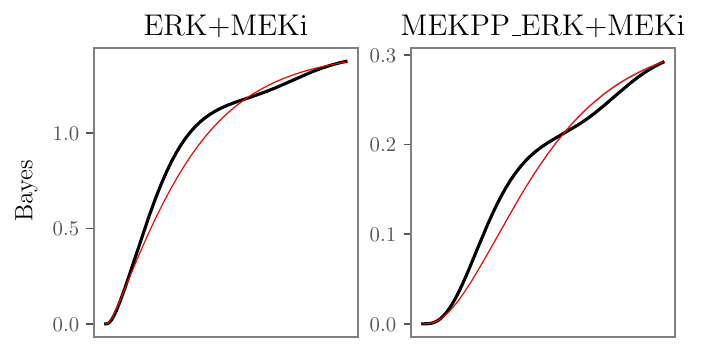} }{ \includegraphics[width=0.3\linewidth]{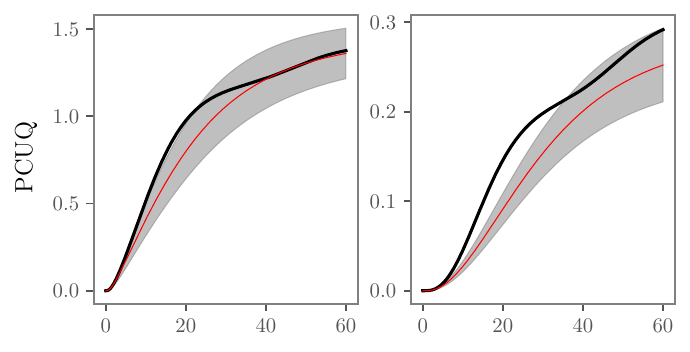} }
    
    \caption{ERK signalling model (centre); this system of 11 \acp{ode} is necessarily misspecified.
    Data (red crosses) were obtained as noisy observations of the system trajectory (black line) in an observational regime (left), and the task was to predict the causal effect of a MEK inhibitor (MEKi) on ERK signalling.
    Predictions (right) produced using standard Bayesian inference (top row) were over-confident, while \ac{pcuq} (bottom row) enabled predictive uncertainty to be accurately quantified.
    (Here the red line indicates predictive mean, while the shaded region indicates predictive quartiles.)
    }
    \label{fig: cell model}
\end{figure*}

Systems biology is ripe for \ac{pcuq}, since the community has invested decades of effort into the design of detailed \ac{ode} descriptions of cellular signalling pathways, with thousands of models hosted on repositories such as \texttt{BioModels} \citep{BioModels2020}.
The sheer complexity of cellular signalling necessarily means that such models are misspecified, yet they continue to serve as valuable tools to communicate understanding and can usefully constrain causal predictions to some extent \citep{alon2019introduction}.
Our contention is that appropriate methodology to harness the predictive performance of such misspecified \acp{ode} within a probabilistic framework does not yet exist, but can be provided by \ac{pcuq}.
To support this argument, we present a case study in which a detailed \ac{ode} description of \ac{erk} signalling is used to predict cellular response to a molecular treatment.

Adopting a similar notation and set-up as \Cref{subsec: ecology}, the evolution of concentrations of molecular species was modelled as the solution $x \mapsto u_\theta(x) \in \mathbb{R}^{d_{\mathcal{Y}}}$ to a system of $d_{\mathcal{Y}} = 11$ \acp{ode}, with parameters collectively denoted $\theta \in \mathbb{R}^p$, $p = 11$.
Data were generated by adding independent Gaussian noise (variance $\sigma^2 I_{11 \times 11}$) to the true molecular concentrations $u(x)$ at discrete times $x_{1:n}$.
The true molecular concentrations $u(x)$ were simulated from the same model, but where the propensity for molecular interaction was no longer constant but time-varying, reflecting the complexity of the cellular environment; full details are contained in \Cref{app: cell}.
For \ac{pcuq}, the kernel $k$, bandwidths, and regularisation parameter $\lambda_n$ were all selected using the same approach used in \Cref{subsec: ecology}.

As the scientific task, we seek to predict the evolution of \ac{erk} under the application of the \ac{meki}, depicted in the centre panel of \Cref{fig: cell model}.
Such interventional data were not part of the training dataset, but as an explicitly causal mechanistic model it is possible to make a causal prediction regarding cell response to treatment\footnote{Briefly, \ac{meki} prevents binding of phosphorylated MEK to ERK (the complex being denoted MEKPP\_ERK), increasing the concentration of free ERK relative to pre-treatment.}.
Predictions made using standard Bayesian inference are contrasted to the predictions made using \ac{pcuq} in \Cref{fig: cell model}.
Failure of Bayesian inference has occurred in an analogous manner to that described in \Cref{subsec: ecology}, with the potentially serious consequence that uncertainty in the efficacy of the \ac{meki} treatment is under-reported (indeed, predictive uncertainty is barely visible in the plot).
In contrast, predictive over-confidence has been avoided through the use of \ac{pcuq}.

%% file: discussion.tex
\section{Discussion}

Deterministic models are an important scientific tool for communicating understanding and insight, but their use in a statistical context typically results in a model that is misspecified.
In a limited data context, the presence of misspecification can often be neglected.
However, model misspecification typically becomes the leading contributor to prediction error as more data are obtained.
Existing (generalised) Bayesian methods are unable to calibrate their predictive confidence in this large data limit, with arbitrarily high confidence associated to a prediction that is incorrect.

This paper proposes a possible solution, called \ac{pcuq}.
The central idea is to replace the usual notion of \emph{average data fit} with a new notion of \emph{predictive fit}.
Operationalising this idea via \ac{mmd} and gradient flows, we demonstrated that catostrophic over-confidence can be avoided in two exemplar \ac{ode} models.
Our work suggests several new directions for future work:
(i) analysing convergence of the particle-discretised gradient flow, following e.g. \cite{chizat22meanfield, hu2021mean, nitanda2022convex}; (ii) obtaining theoretical insight into the selection of $\lambda_n$; and (iii) obtaining guarantees on the coverage of $Q_n$ when the model is well-specified, in particular accounting for the plug-in approximation \eqref{eq: joint disn} in the case where data are dependent.

The issue of a misspecified deterministic model is not itself new, and methods that attempt to \emph{learn the residual} date back at least to \citet{kennedy2001bayesian}.
These methods are also prediction-centric, but they sacrifice any causal semantics associated to the statistical model by modifying the model itself, through the introduction of a nonparametric component.
A key distinguishing feature of \ac{pcuq} is that, while being prediction-centric, causal semantics associated to the model are retained.
On the negative side, this means that (possibly causal) predictions made using \ac{pcuq} are ultimately gated by the performance of the original statistical model.
This highlights the scale of the challenge associated with model misspecification; an area where substantial further work is required.

%% file: acknowledgements.tex
\subsubsection*{Acknowledgements}
The authors are grateful to Joshua Bon, Aidan Mullins, and Veit Wild for discussion of earlier work.
ZS and CJO were supported by EP/W019590/1.
CJO acknowledges support from The Alan Turing Institute in the UK.
JK was supported through the UK’s Engineering and Physical Sciences Research
Council (EPSRC) via EP/W005859/1 and EP/Y011805/1.

%% file: appendix.tex
{\noindent \Large \textbf{Supplement}}

\noindent These appendices accompany the manuscript \textit{Prediction-Centric Uncertainty Quantification}.

\section{Existence, Uniqueness, and Characterisation of $Q_n$}
\label{ap: theory}

This appendix presents sufficient conditions for the existence and uniqueness of a minimiser $Q_n$ of \eqref{eq:prediction-centric-uq}, and an implicit characterisation of $Q_n$ which in turn suggests some alternative sampling strategies that could be used (c.f. \Cref{ap: mcmc}).

First, we establish sufficient conditions for existence and uniqueness of $Q_n$, which are stated in \Cref{thm: exist}.
These leverage the convexity of \eqref{eq: free energy}, which follows from entropic regularisation; i.e. from the use of \ac{kl} as the regulariser in \eqref{eq:prediction-centric-uq}.

\begin{theorem}[Existence and Uniqueness of $Q_n$]
\label{thm: exist}
Assume that $\Theta = \mathbb{R}^p$, $p \in \mathbb{N}$, and that:
\begin{enumerate}
    \item $k : \mathcal{X} \times \mathcal{X} \rightarrow \mathbb{R}$ is bounded
    \item $Q_0$ has density $q_0 \propto \exp( - U_0)$ where $\nabla U_0$ is Lipschitz and satisfies a coercivity condition $\langle \theta, \nabla U_0 (\theta) \rangle \geq  c\cdot \| \theta \|^2 - c^\prime $ for some $c > 0,c' \in \mathbb{R}$, and all $\theta \in \Theta$.
\end{enumerate}
    Then $Q_n$ exists and is unique for all $\lambda_n > 0$.
\end{theorem}
\begin{proof}[Proof of \Cref{lem: characterise Qn}]
    The result is essentially that of Proposition 2.5 of \cite{hu2021mean}; here for the reader we provide the high-level argument.
    The starting point is to recall that $Q_n$ minimises \eqref{eq:prediction-centric-uq}; i.e. $Q_n$ minimises
    \begin{align*}
        \mathcal{F}_n(Q) := \frac{1}{2} \operatorname{MMD}^2( P_n , P_{Q} ) + \operatorname{KL}(Q, Q_0) 
    \end{align*}
    Since we have assumed that our kernel $k$ is bounded, it follows that the $\operatorname{MMD}$ term is also bounded. 
    It hence follows that $\mathcal{F}_n$ is finite whenever $\mathrm{KL} \left( Q, Q_0 \right)$ is finite, i.e. the domain of $\mathcal{F}_n$ is the domain of $\mathrm{KL} \left( \cdot, Q_0 \right)$. It is well-known that this functional is strictly convex on its convex domain, and bounded below by $0$. Additionally, for any Mercer kernel $k$, the functional $\operatorname{MMD}^2$ is always convex (classically, rather than geodesically), and the mapping $Q \mapsto P_{Q}$ is linear, implying that the composite mapping $Q \mapsto \operatorname{MMD}^2( P_n , P_{Q} )$ is also convex and lowe-bounded. It hence follows that $\mathcal{F}_n$ is strictly convex on its convex domain, and is lower-bounded. Existence of a minimiser holds by routine arguments involving lower semi-continuity of the objective, see e.g. Appendix A of \cite{wild2023rigorous}, and strict convexity implies uniqueness. 
\end{proof}

The second result we present is an implicit characterisation for $Q_n$ itself, in \Cref{lem: characterise Qn}.
This is obtained by setting to zero the Wasserstein gradient $\nabla_{\mathrm{W}} \mathcal{F}_n(Q)$.

\begin{theorem}[Characterisation of $Q_n$]
\label{lem: characterise Qn}
    Under the same conditions as \Cref{thm: exist}, $Q_n$ satisfies the implicit equation
    \begin{align}
        Q_n \left( \mathrm{d} \theta \right) & = \frac{1}{Z} \cdot Q_0 ( \mathrm{d} \theta ) \cdot \exp \left( - \frac{1}{\lambda_n} \cdot V_{Q_n} ( \theta ) \right) \label{eq: implicit Qn} 
    \end{align}
    where, letting $\kappa_0(\theta,\vartheta) = \langle \mu_k(P_\theta) , \mu_k(P_\vartheta) \rangle_{\mathcal{H}(k)}$,
    \begin{align*}
        V_{Q_n}(\theta) &= \int \kappa_0(\theta, \cdot)\mathrm{d}Q_n - \langle \mu_k(P_n), \mu_k(P_\theta)\rangle_{\mathcal{H}(k)} ,
    \end{align*}
    and where $Z > 0$ is a normalisation constant.
\end{theorem}
\begin{proof}
    Using the expression for the Stein kernel $\kappa_{P_n}$ in \eqref{eq: explicit mmd}, and the fact that $Q$ is a probability measure, we can write the entropy-regularised free energy \eqref{eq: free energy} as
    \begin{align*}
        \mathcal{F}_n(Q) & = \frac{1}{2} \left\{ \frac{1}{n^2} \sum_{i=1}^n \sum_{j=1}^n k(y_i,y_j) - \frac{2}{n} \sum_{i=1}^n \iint k(y_i,y) \; \mathrm{d}P_\theta(y) \; \mathrm{d}Q(\theta) \right. \\
        & \left. \hspace{80pt}  + \iiiint k(y,y') \; \mathrm{d}P_\theta(y) \mathrm{d} P_\vartheta(y') \; \mathrm{d}Q(\theta) \mathrm{d}Q(\vartheta) \right\}  + \lambda_n \int \log \left( \frac{\mathrm{d}Q}{\mathrm{d}Q_0}(\theta) \right) \; \mathrm{d}Q(\theta)
    \end{align*}
    Taking the functional derivative with respect to $Q$, we obtain that
    \begin{align*}
        \delta \mathcal{F}_n(Q)(\theta) & = \underbrace{ - \frac{1}{n} \sum_{i=1}^n \int k(y_i,y) \; \mathrm{d}P_\theta(y) + \iiint k(y,y') \; \mathrm{d}P_\theta(y) \mathrm{d} P_\vartheta(y') \; \mathrm{d}Q(\vartheta) }_{= V_{Q_n(\theta)}} \; + \; \lambda_n \left[ 1 + \log \left( \frac{\mathrm{d}Q}{\mathrm{d}Q_0}(\theta) \right) \right] .
    \end{align*}
    The minimiser $Q = Q_n$ of the free energy satisfies $\delta \mathcal{F}_n(\theta) = \mathrm{constant}$ (where the $\mathrm{constant}$ term can be seen as a Lagrange multiplier reflecting the constraint that $Q_n$ be a probability measure), leading to the implicit equation
    \begin{align*}
        0 = V_{Q_n}(\theta) + \lambda_n \log \left( \frac{\mathrm{d}Q}{\mathrm{d}Q_0}(\theta) \right) + \mathrm{constant},
    \end{align*}
    and upon rearranging we obtain
    \begin{align*}
        \frac{\mathrm{d}Q}{\mathrm{d}Q_0}(\theta) &\propto \exp\left( - \frac{1}{\lambda_n} V_{Q_n}(\theta) \right) ,
    \end{align*}
    which is equivalent to the stated result.
\end{proof}

While the regularising effect of $Q_0$ is well-understood for Gibbs measures like $Q_n^{\dagger}$ \citep{bissiri2016general, knoblauch2022optimization},  it is perhaps surprising that $Q_0$ acts on $Q_n$ in essentially the same way in \ac{pcuq}, as shown in \eqref{eq: implicit Qn}.
This implicit characterisation of $Q_n$ is also interesting in that it suggests alternative sampling strategies, which are discussed in \Cref{ap: mcmc}.

\section{Computing the Stein Kernel $\kappa_{P_n}$ and its Gradient $\nabla_1 \kappa_{P_n}$}
\label{subsec: compute gradients}

This appendix is devoted to discussion of different computational strategies to computing integrals that appear in both the Stein kernel $\kappa_{P_n}$ and its gradient $\nabla_1 \kappa_{P_n}$.
The case of independent data is discussed in \Cref{app: scores fo iid}, while the case of dependent data is discussed in \Cref{app: gradients for non IID}.
In addition, we highlight that a reduction in computational complexity is possible for dependent data in the regime $\ell_{\mathcal{X}} \rightarrow 0$ in \Cref{app: ell x zero limit}.

\subsection{Independent Data}
\label{app: scores fo iid}

Simulation of \eqref{eq: wgd_mv_ip} requires the gradient $\nabla_1 \kappa_{P_n}$ of the Stein kernel.
Several techniques are available to compute the gradient (with respect to $\theta$) of the integrals appearing in expression \eqref{eq: explicit mmd} for the Stein kernel.
For simplicity we focus on the case of integrals with respect to $P_\theta$, but double integrals with respect to $P_\theta$ can be similarly handled.

\paragraph{Analytic Case}

For particular combinations of $k$ and $P_\theta$ the integrals appearing in \eqref{eq: explicit mmd} will be computable in closed form and derivatives of the kernel $\kappa_{P_n}$ may be exactly computed.
This scenario occurs, for example, when $k$ is a Gaussian kernel and $P_\theta$ is a Gaussian distribution; other scenarios where the kernel mean embedding can be exactly computed are listed in Table 1 of \citet{briol2019probabilistic}.

\paragraph{Score Gradient Case}

In the case where $P_\theta$ admits a positive and differentiable pdf $p_\theta(\cdot)$, under regularity conditions we can calculate that
\begin{align}
    \nabla_\theta \int k(y_i,y) \; \mathrm{d}P_\theta(y) & = \nabla_\theta \int k(y_i,y) p_\theta(y) \; \mathrm{d}y 
    = \int k(y_i,y) \nabla_\theta p_\theta(y) \; \mathrm{d}y  
    = \int k(y_i , y) (\nabla_\theta \log p_\theta(y)) \; \mathrm{d}P_\theta(y) \label{eq: score gradient}
\end{align}
and obtain a natural approximation using Monte Carlo.
This is analogous to how gradients are calculated in \emph{black-box variational inference} \citep{ranganath2014black}.

\paragraph{Reparametrisation Trick}

If we can express $P_\theta = \mathsf{law}(f_\theta(U))$ for some $\theta$-independent random variable $U \sim \mathcal{U}$ then we can employ the \emph{reparametrisation trick} to express
\begin{align*}
    \nabla_\theta \int k(y_i,y) \; \mathrm{d}P_\theta(y) & = \nabla_\theta \int k(y_i, f_\theta(u) ) \; \mathrm{d}\mathcal{U}(u) 
    = \int \nabla_\theta f_\theta(u) \nabla_2 k(y_i,f_\theta(u))  \; \mathrm{d}\mathcal{U}(u) 
\end{align*}
and obtain a natural approximation using Monte Carlo.
This is analogous to how gradients are calculated in \emph{automatic differentiation variational inference} \citep{kucukelbir2017automatic}.

\subsection{Dependent Data}
\label{app: gradients for non IID}

The strategies for computing the gradient $\nabla_1 \kappa_{P_n}$ that were presented in \Cref{app: scores fo iid} for the setting where data are independent can be extended to the regression setting considered in \Cref{subsec: non IID}, where data are dependent.

\paragraph{Analytic Case}

The Stein kernel is now
\begin{align}
\kappa_{P_n}(\theta,\vartheta) & = \frac{1}{n^2} \sum_{i,j=1}^n k((x_i,y_i),(x_j,y_j)) - \frac{1}{n^2} \sum_{i,j=1}^n \int k((x_i,y_i),(x_j,y)) \; \mathrm{d}P_\theta(y|x_j) \label{eq: extended kappa} \\
& \qquad - \frac{1}{n^2} \sum_{i,j=1}^n \int k((x_i,y_i),(x_j,y)) \; \mathrm{d}P_\vartheta(y|x_j) + \frac{1}{n^2} \sum_{i,j=1}^n \iint k((x_i,y),(x_j,y')) \; \mathrm{d}P_\theta(y|x_i) \mathrm{d} P_\vartheta(y'|x_j) . \nonumber
\end{align}
The case where $k : (\mathcal{X} \times \mathcal{Y}) \times (\mathcal{X} \times \mathcal{Y}) \rightarrow \mathbb{R}$ is a separable Gaussian kernel and $P_\theta(\cdot|x)$ is a Gaussian measurement error model can be calculated in closed form, and we present this calculation in \Cref{app: gaussian calculations}.
This was the computational approach used to perform all experiments in \Cref{sec: experiments}.

\paragraph{Score Gradient Case}

The score gradient approach can be immediately extended to the dependent data setting via
\begin{align}
    \nabla_\theta \int k((x_i,y_i),(x,y)) \; \mathrm{d}\bar{P}_\theta(x,y) 
    & = \nabla_\theta \left\{ \frac{1}{n} \sum_{j=1}^n \int k((x_i,y_i),(x_j,y)) \; \mathrm{d}P_\theta(y|x_j) \right\} \nonumber \\
    & = \frac{1}{n} \sum_{j=1}^n \int k((x_i,y_i),(x_j,y)) (\nabla_\theta \log p_\theta(y | x_j)) \; \mathrm{d}P_\theta(y|x_j) .   \label{eq: extend score gradient}
\end{align}

\paragraph{Reparametrisation Trick}

For the reparametrisation trick we now require a map $f_\theta$ such that $y_i | x_i$ is modelled as $f_\theta(U,x)$ where $U \sim \mathcal{U}$.
The gradient that we seek is then
\begin{align*}
     \nabla_\theta \int k((x_i,y_i),(x,y)) \; \mathrm{d}\bar{P}_\theta(x,y) & = \nabla_\theta \left\{ \frac{1}{n} \sum_{j=1}^n \int k((x_i,y_i),(x_j,f_\theta(u,x_j))) \; \mathrm{d}\mathcal{U}(u)  \right\} \\
    & = \frac{1}{n} \sum_{j=1}^n \int \nabla_\theta f_\theta(u,x_j) \nabla_{2,2} k((x_i,y_i),(x_j,f_\theta(u,x_j))) \; \mathrm{d}\mathcal{U}(u)
\end{align*}
where $\nabla_{2,2} k((a,b),(c,d))$ denotes the gradient with respect to argument $d$.

\subsection{Dependent Data; Simplification when $\ell_{\mathcal{X}} \rightarrow 0$}
\label{app: ell x zero limit}

In the dependent data setting, evaluation of the Stein kernel $\kappa_{P_n}$ in \eqref{eq: extended kappa} incurs a computational cost of $O(n^2)$ in general, as opposed to the $O(n)$ cost of computing \eqref{eq: explicit mmd} when data are independent.
Fortunately we can simultaneously recover the $O(n)$ cost while also mitigating negative effects of the lifting in \Cref{subsec: non IID} by considering the $(x,x')$ dependence of the kernel $k((x,y),(x',y'))$ in a particular objective limit.
This follows an identical strategy proposed in Section 3.4 of \citet{alquier2024universal}.

For the experiments that we report in the main text we used the Gaussian kernel where $\ell_{\mathcal{X}}$ is the bandwidth describing the similarity between two distinct covariates $x$ and $x'$.
A natural limit exists when $\ell_{\mathcal{X}} \rightarrow 0$, where  upon the kernel $k : (\mathcal{X} \times \mathcal{Y}) \times (\mathcal{X} \times \mathcal{Y}) \rightarrow \mathbb{R}$ becomes
\begin{align}
k((x,y),(x',y')) = \exp\left(- \frac{\|x-x'\|^2}{\ell_{\mathcal{X}}^2} - \frac{\|y-y'\|^2}{\ell_{\mathcal{Y}}^2} \right) \rightarrow \left\{ \begin{array}{ll} k(y,y') := \exp\left( - \frac{\|y-y'\|^2}{\ell_{\mathcal{Y}}^2} \right) & \text{if } x = x' \\ 0 & \text{if } x \neq x' \end{array} \right.
\end{align}
and upon plugging this limiting expression into \eqref{eq: extended kappa} we obtain
\begin{align*}
\kappa_{P_n}(\theta,\vartheta) & = \frac{1}{n^2} \sum_{i=1}^n k(y_i,y_i) - \frac{1}{n^2} \sum_{i=1}^n \int k(y_i,y) \; \mathrm{d}P_\theta(y|x_i) \\
& \qquad - \frac{1}{n^2} \sum_{i=1}^n \int k(y_i,y) \; \mathrm{d}P_\vartheta(y|x_i) + \frac{1}{n^2} \sum_{i=1}^n \iint k(y,y') \; \mathrm{d}P_\theta(y|x_i) \mathrm{d} P_\vartheta(y'|x_i) ,
\end{align*}
which is seen to have computational cost $O(n)$.
In addition to reducing cost, this limit has the effect that \emph{only} data $y_i$ corresponding to precisely the covariate $x_i$ are used to assess the performance of the model $P_\theta(\mathrm{d}y | x_i)$, mitigating negative effects of the lifting that we described in \Cref{subsec: non IID}.

\section{Alternative Sampling Methods for $Q_n$}
\label{ap: mcmc}

The aim of this appendix is to comment on the possibility of alternative sampling strategies for $Q_n$.
From \eqref{eq: wgd_mv_ip}, one can verify by inspection that $\{\theta_t\}_{1:N} = (\theta_t^1 , \dots , \theta_t^N)$ is precisely the Langevin diffusion with joint target distribution
\begin{align}
    \tilde{Q}_n^{\otimes N} \left( \{ \mathrm{d} \theta\}_{1:N}\right) \propto \left( \prod_{i = 1}^N Q_0 (\mathrm{d} \theta^i) \right) \cdot \exp \left( - \frac{1}{\lambda_n} \cdot \frac{1}{N-1} \sum_{1 \leq i < j \leq N} \kappa_{P_n}(\theta^i, \theta^j) \right)
    \label{eq:joint_density}
\end{align}
Following this observation, one could consider using a general \ac{mcmc} sampler to generate approximate samples from \eqref{eq:joint_density}; at stationarity, each of the $N$ components of the Markov chain would then represent an approximate sample from $Q_n$.
Of course, for large $N$ the dimension of the state space $\Theta^N$ can be expected to create difficulties for efficient \ac{mcmc} schemes which exactly preserve the target measure. Nevertheless, for approximate algorithms such as the unadjusted Langevin algorithm, implementation remains feasible, and there are even suggestions \citep{durmus2021asymptotic, chen2024convergence} that the asymptotic bias remains well-behaved in this limit; we highlight this as a direction for future work.

Comparing the above with $Q_n^{\otimes N}(\{\theta\}_{1:N}):=\prod_{i=1}^N Q_n(\theta^i)$, the joint density of $N$ independent draws from $Q_n$ \eqref{eq: implicit Qn}, we observe that $\tilde{Q}_n^{\otimes N}$ \eqref{eq:joint_density} differs in that the implicit term $\int\kappa_0(\theta, \cdot)\mathrm{d}Q_n$ is replaced by an explicit term $\frac{1}{N-1}\sum_{j\ne i} \kappa_0(\theta^i, \theta^j)$. 
This is essentially the key ingredient of mean-field Langevin dynamics; see e.g. \cite{chizat22meanfield, hu2021mean, nitanda2022convex}.

\section{Calculations for Gaussian Measurement Models}
\label{app: gaussian calculations}

This appendix contains explicit calculations for the integral terms appearing in the Stein kernel $\kappa_{P_n}$ and its gradient $\nabla_1 \kappa_{P_n}$ for the Gaussian measurement error model $P_\theta(\cdot|x) = \mathcal{N}(w_\theta(x),\sigma^2 I_{p \times p})$, which is extremely common in scientific modelling.
For analytic tractability we consider a separable Gaussian kernel
\begin{align*}
    k((x,y),(x',y')) = \exp\left( - \frac{\|x-x'\|^2}{2 \ell_{\mathcal{X}}^2} \right) \exp \left( - \frac{\|y-y'\|^2}{2 \ell_{\mathcal{Y}}^2} \right)
\end{align*}
so that we can analytically evaluate the integrals
\begin{align*}
    \int \exp \left( - \frac{\|y_i - y \|^2}{2 \ell_{\mathcal{Y}}^2} \right) \; \mathrm{d}P_\theta(y|x_j) & = \left( \frac{\ell_{\mathcal{Y}}^2}{\ell_{\mathcal{Y}}^2 + \sigma^2} \right)^{p/2} \exp\left( - \frac{ \| y_i - w_\theta(x_j) \|^2 }{2(\ell_{\mathcal{Y}}^2 + \sigma^2)} \right) \\
    \iint \exp \left( - \frac{\|y - y' \|^2}{2 \ell_{\mathcal{Y}}^2} \right) \; \mathrm{d}P_\theta(y|x_i) \mathrm{d}P_\vartheta(y'|x_j) & = \left( \frac{\ell_{\mathcal{Y}}^2}{\ell_{\mathcal{Y}}^2 + 2 \sigma^2} \right)^{p/2} \exp\left( - \frac{ \| w_\theta(x_i) - w_\vartheta(x_j) \|^2 }{2(\ell_{\mathcal{Y}}^2 + 2\sigma^2)} \right)
\end{align*}
and explicitly differentiate these with respect to $\theta$ to obtain
\begin{align*}
    \nabla_\theta \int \exp \left( - \frac{\|y_i - y \|^2}{2 \ell_{\mathcal{Y}}^2} \right) \; \mathrm{d}P_\theta(y|x_j) & = \frac{1}{\ell_{\mathcal{Y}}^2} \left( \frac{\ell_{\mathcal{Y}}^2}{\ell_{\mathcal{Y}}^2 + \sigma^2} \right)^{p/2 + 1} \\
    & \hspace{10pt} \times \exp\left( - \frac{ \| y_i - w_\theta(x_j) \|^2 }{2(\ell_{\mathcal{Y}}^2 + \sigma^2)} \right) [\nabla_\theta w_\theta(x_j)] [y_i - w_\theta(x_j)] \\
    \nabla_\theta \iint \exp \left( - \frac{\|y - y' \|^2}{2 \ell_{\mathcal{Y}}^2} \right) \; \mathrm{d}P_\theta(y|x_i) \mathrm{d}P_\vartheta(y'|x_j) & = \frac{1}{\ell_{\mathcal{Y}}^2} \left( \frac{\ell_{\mathcal{Y}}^2}{\ell_{\mathcal{Y}}^2 + 2 \sigma^2} \right)^{p/2 + 1} \\
    & \hspace{10pt} \times \exp\left( - \frac{ \| w_\theta(x_i) - w_\vartheta(x_j) \|^2 }{2(\ell_{\mathcal{Y}}^2 + 2\sigma^2)} \right) [ - \nabla_\theta w_\theta(x_i)] [w_\theta(x_i) - w_\vartheta(x_j)] .
\end{align*}

\section{Additional Empirical Details}

This appendix contains all details needed to reproduce the empirical results that we report in \Cref{sec: experiments}.
Details for the \ac{lvm} experiment are contained in \Cref{app: extra experiments}, while details for the cell signalling experiment are contained in \Cref{app: cell}.

\subsection{Lotka--Volterra Model}
\label{app: extra experiments}

\paragraph{Parameter Settings}

The data-generating parameters that we used for the \ac{lvm} were:
$\alpha = \mathrm{logit}^{-1}(-2)$, $\beta = \mathrm{logit}^{-1}(-4)$, $\gamma = 0.4$, $\delta = 0.02$.
The initial prey and predator concentrations were $\xi_1 = 10$ and $\xi_2 = 10$.
The \ac{sde} model described in \Cref{subsec: ecology} with intrinsic noise $\epsilon_1 = 0$, $\epsilon_2 = 0.4$ was discretised for numerical simulation using the reversible Heun's method with time step $\mathrm{d}x = 0.01$, and simulated from $x = 0$ to $x = 60$.
Data were extracted at times $x_{1:n}$ ranging from $0$ to $60$ in increments of $1.0$, with Gaussian measurement noise of variance $\sigma^2 = 1$ added.

\paragraph{MMD-Bayes}

The MMD-Bayes method of \citet{cherief2020mmd} is an instanced of generalised variational inference whose output distribution $Q$ is given by the Radon--Nikodym derivative
\begin{align*}
    \frac{\mathrm{d}Q}{\mathrm{d}Q_0}(\theta) = \exp\left( - \beta \mathrm{MMD}^2(P_n,P_\theta) \right)
\end{align*}
where $\beta > 0$ was termed a \emph{learning rate}.
For our implementation of MMD-Bayes we used the same kernel $k$ to construct the \ac{mmd} as was used for \ac{pcuq}, and we employed the same learning rate $\beta = \exp(np)$ as used in \citet{cherief2020mmd} where $\theta \in \mathbb{R}^p$ ($p = 2$ for the \ac{lvm}).

\paragraph{Convergence of Sampling Methods}

The standard Bayesian posterior and the MMD-Bayes output were each numerically approximated using 10 parallel instances of the \ac{mala} \citep{roberts1996exponential}, initialised from the true parameter values (i.e. to ensure a warm start).
A total of 5000 iterations were performed, with the proposal variance parameter tuned to avoid pathological behaviour, and only the remaining third of the samples were retained. The \ac{pcuq} output was obtained using \ac{mala} on the joint distribution of 10 particles, together with the warm start with the same number of iterations.
Trace plots, demonstrating convergence of the sampling algorithms, are displayed in \Cref{fig: lv3}.

\begin{figure}[t]
\centering
    \includegraphics[width=0.8\textwidth]{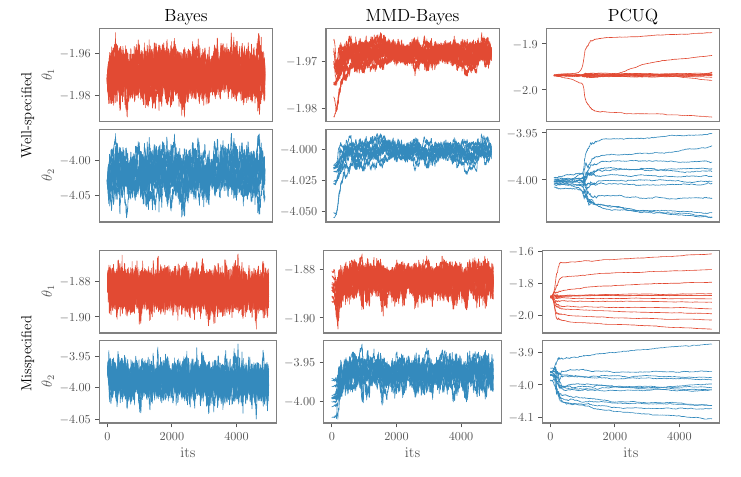}
    \caption{Lotka--Volterra model: Convergence of sampling algorithms. Here trace plots for $\theta_1$ (red) and $\theta_2$ (blue) are presented for the sampling algorithms used to approximate the standard Bayesian posterior (left), the MMD-Bayes output (middle) and the proposed \ac{pcuq} method (right).  The top row corresponds to the well-specified case where data were generated from the \ac{ode} model, while the bottom row corresponds to the misspecified case where data were generated from the \ac{sde} model.   }
    \label{fig: lv3}
\end{figure}

\paragraph{Insensitivity to $\lambda_\star$}

The proposed \ac{pcuq} method involves a regularisation parameter $\lambda_n$, and in the main text we proposed a heuristic choice $\lambda_n = \lambda_\star$. To assess the sensitivity to this choice, we re-computed predictive distributions based on alternative values of $\lambda_n$, presenting the results in \Cref{fig: lv2}.
These results demonstrate that predictive performance is rather insensitive to reducing the size of $\lambda_n$ (of course, for very large $\lambda_n$ we will recover $Q_0$).
This may be explained by the fact that predictive uncertainty is driven mainly by the data in \ac{pcuq}, with the role of the distribution $Q_0$ being limited to convexification of the gradient flow objective in \Cref{eq:prediction-centric-uq}, and for this purpose a relatively small value of $\lambda_n$ is sufficient.

\begin{figure}[t]
\centering
    \includegraphics[width=\textwidth]{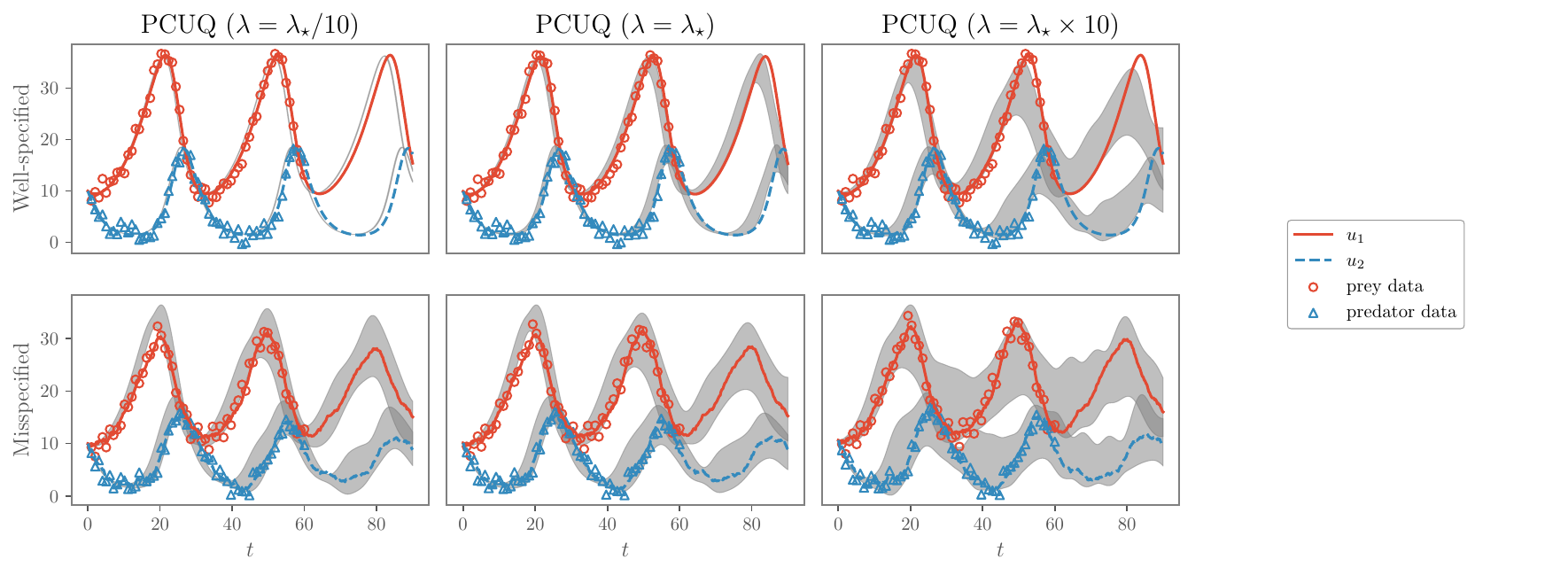}
    \caption{Lotka--Volterra Model:  Investigating sensitivity to the regularisation hyperparameter $\lambda_\star$.  The proposed \ac{pcuq} method involves a regularisation parameter $\lambda_n$, and in the main text we proposed a heuristic choice $\lambda_n = \lambda_\star$. To assess the sensitivity to this choice, we re-computed predictive distributions based on the alternative choices $\lambda_n = \lambda_\star / 10$ (left) and $\lambda_n = \lambda_\star \times 10$ (right).  It was observed that predictions were almost unchanged in both the well-specified and misspecified context.}
    \label{fig: lv2}
\end{figure}

\paragraph{Insensitivity to $N$}

The proposed gradient flow algorithm involves selecting a number of particles $N$, and in the main text we presented results based on $N = 10$. 
To assess the sensitivity to this choice, we re-computed trace plots based on alternative choices of $N$ in \Cref{fig: lv4,fig: lv5}.
A greater number of particles provides a more accurate spatial discretisation on the right hand side of \eqref{eq: wgd_mv_ip}, as well as better expressive flexibility of characterising uncertainty. 
In the well-specified case (\Cref{fig: lv4}), we see that that the number of particles affects the discretisation of $Q^t$ when the number of particles is low.
In the misspecified case (\Cref{fig: lv5}), we see that despite the increasing number of particles, the overall span of the particle distribution remains unchanged. 
Therefore, the outcome of \ac{pcuq} remains insensitive to the number of particles when the discretization and expressivity are adequately addressed. 

\begin{figure}[t]
\centering
    \includegraphics[width=0.8\textwidth]{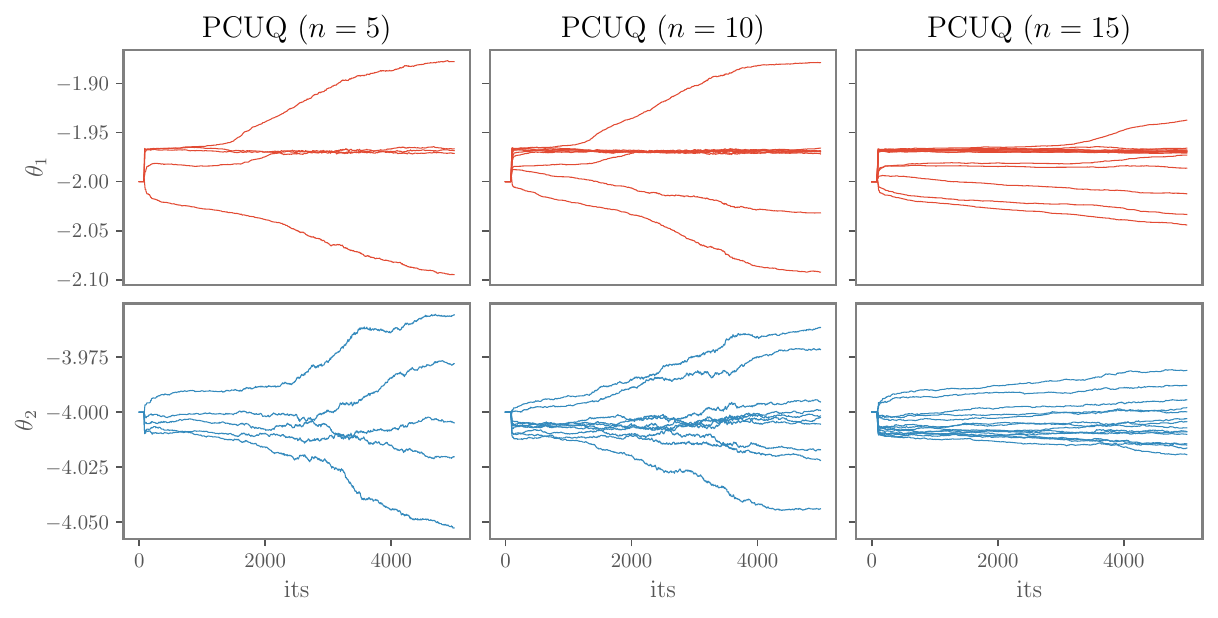}
    \caption{Lotka--Volterra Model:  Investigating sensitivity to the number of particles $N$.  
    The proposed gradient flow algorithm involves selecting a number of particles $N$, and in the main text we presented results based on $N = 10$. 
    To assess the sensitivity to this choice, we re-computed trace plots based on the alternative choices $N=5$ (left) and $N=15$ (right).
    Here data were generated from the well-specified model. }
    \label{fig: lv4}
\end{figure}

\begin{figure}[t]
\centering
    \includegraphics[width=0.8\textwidth]{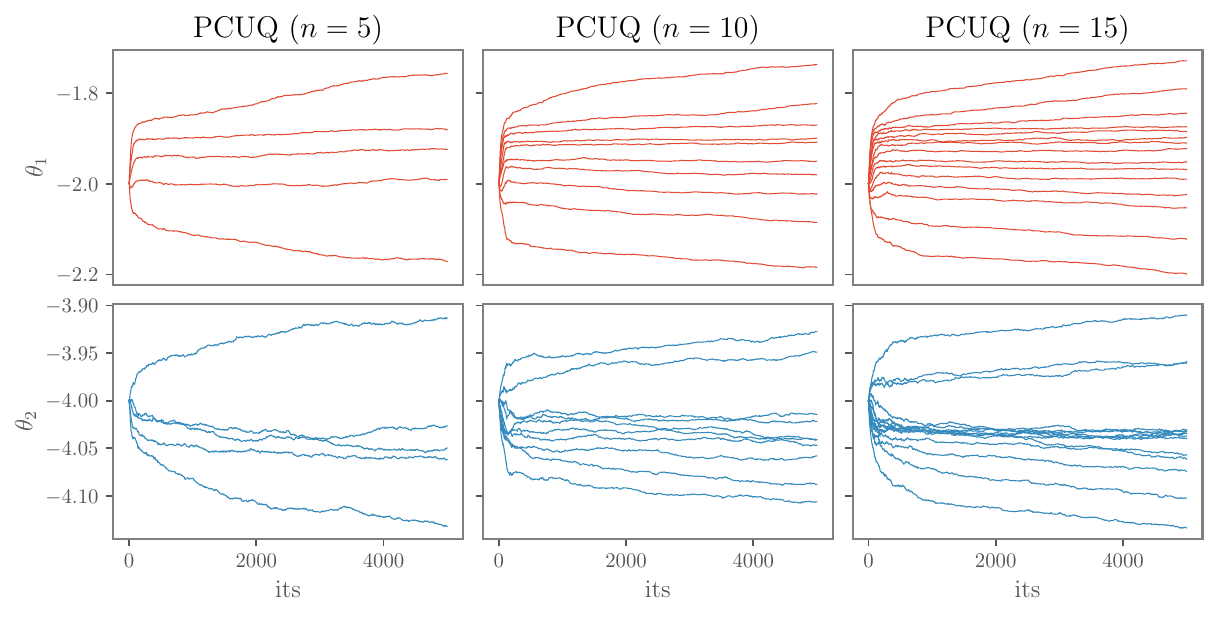}
    \caption{Lotka--Volterra Model:  Investigating sensitivity to the number of particles $N$.  
    The proposed gradient flow algorithm involves selecting a number of particles $N$, and in the main text we presented results based on $N = 10$. 
    To assess the sensitivity to this choice, we re-computed trace plots based on the alternative choices $N=5$ (left) and $N=15$ (right).
    Here data were generated from the misspecified model. }
    \label{fig: lv5}
\end{figure}

\paragraph{Dependence on Parameter Settings}

In response to a reviewer who asked about the sensitivity of these results to the specific parameter settings, we repeated the \ac{lvm} experiment with the alternative parameter settings $\alpha = \mathrm{logit}^{-1}(-1)$, $\beta = \mathrm{logit}^{-1}(-3)$, $\gamma = 0.4$, $\delta = 0.02$.
The initial prey and predator concentrations this time were $\xi_1 = 10$ and $\xi_2 = 15$.
The \ac{sde} model described in \Cref{subsec: ecology} was employed with intrinsic noise $\epsilon_1 = 0.1$, $\epsilon_2 = 0.2$.
The corresponding analogue of \Cref{fig: lv} for these alternative parameter settings is shown as \Cref{fig: lv mod}.
Likewise \Cref{fig: lv3 mod,fig: lv4 mod,fig: lv5 mod} are the analogues of \Cref{fig: lv3,fig: lv4,fig: lv5}.

\begin{figure*}[t]
\centering
    \includegraphics[width=\textwidth]{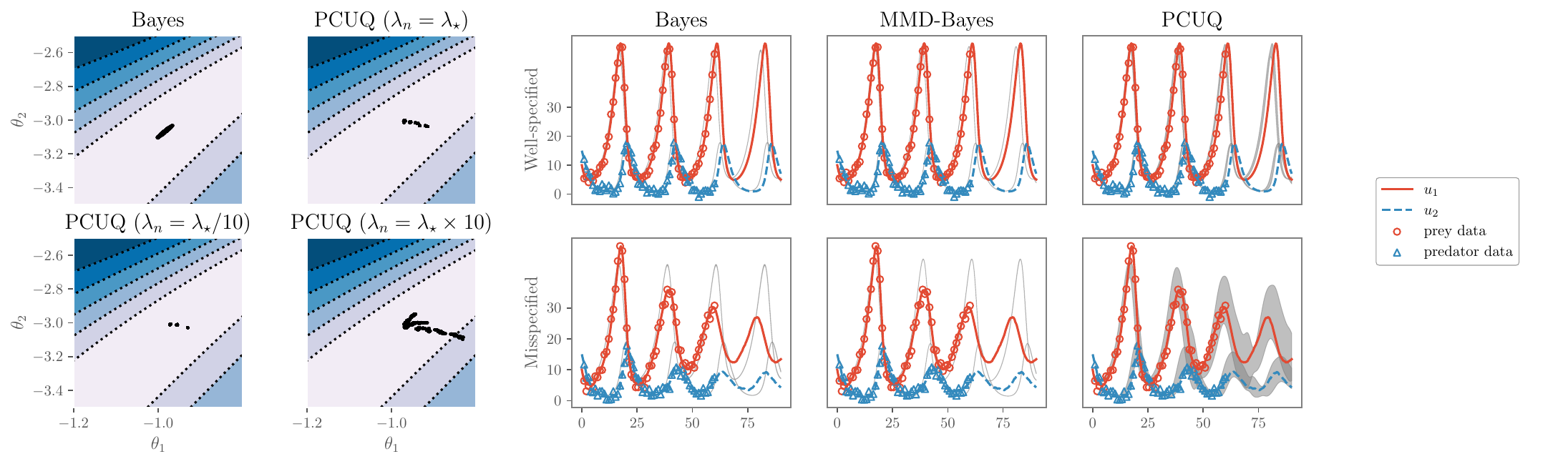}
    \caption{Lotka--Volterra model (alternative parameter settings).  Left:  Contour plots depict the standard Bayesian posterior, superimposed with samples from the standard Bayesian posterior (Bayes) and the proposed method (PCUQ) with varying regularisation parameter $\lambda_n$.
    Right:  Predictive distributions in both the well-specified (top) and misspecified (bottom) context.
    (Lines depict true prey and predator populations, while the shaded regions depict the predictive quartiles for standard Bayesian inference, the MMD-Bayes method of \citet{cherief2020mmd}, and the proposed \ac{pcuq}.)
    }
    \label{fig: lv mod}
\end{figure*}

\begin{figure}[t]
\centering
    \includegraphics[width=0.8\textwidth]{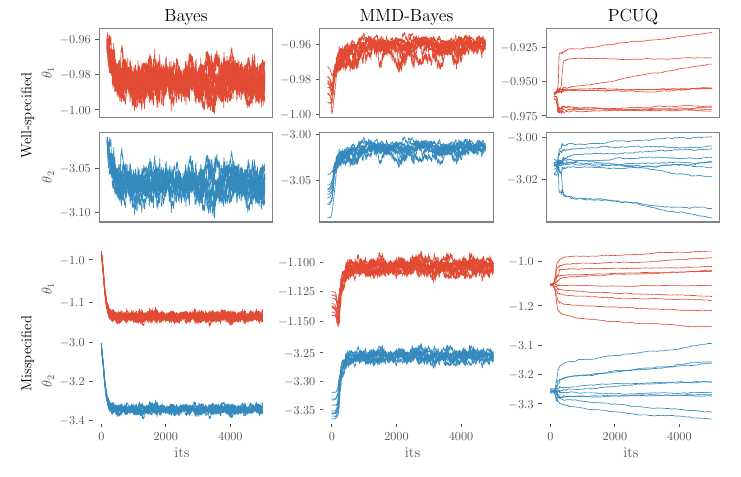}
    \caption{Lotka--Volterra model (alternative parameter settings). Convergence of sampling algorithms. Here trace plots for $\theta_1$ (red) and $\theta_2$ (blue) are presented for the sampling algorithms used to approximate the standard Bayesian posterior (left), the MMD-Bayes output (middle) and the proposed \ac{pcuq} method (right).  The top row corresponds to the well-specified case where data were generated from the \ac{ode} model, while the bottom row corresponds to the misspecified case where data were generated from the \ac{sde} model.   }
    \label{fig: lv3 mod}
\end{figure}

\begin{figure}[t]
\centering
    \includegraphics[width=0.8\textwidth]{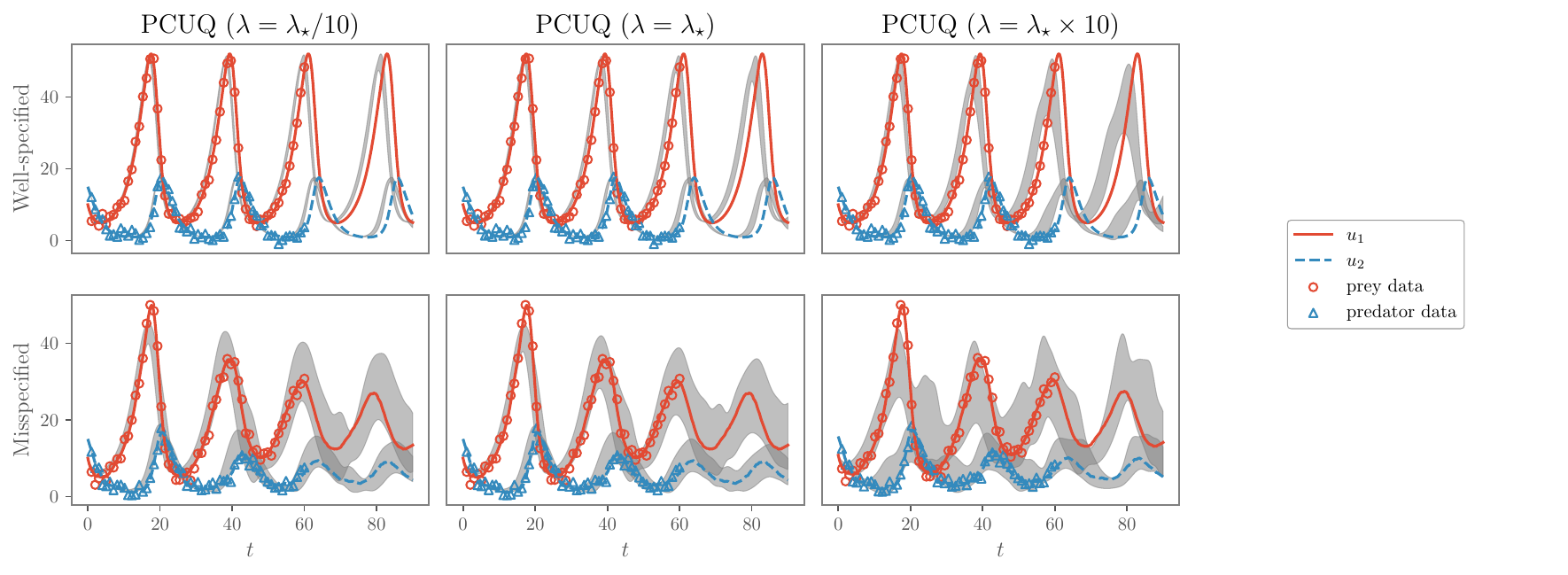}
    \caption{Lotka--Volterra Model (modified parameter settings).  
    Investigating sensitivity to the number of particles $N$.  
    The proposed gradient flow algorithm involves selecting a number of particles $N$, and in the main text we presented results based on $N = 10$. 
    To assess the sensitivity to this choice, we re-computed trace plots based on the alternative choices $N=5$ (left) and $N=15$ (right).
    Here data were generated from the well-specified model. }
    \label{fig: lv4 mod}
\end{figure}

\begin{figure}[t]
\centering
    \includegraphics[width=0.8\textwidth]{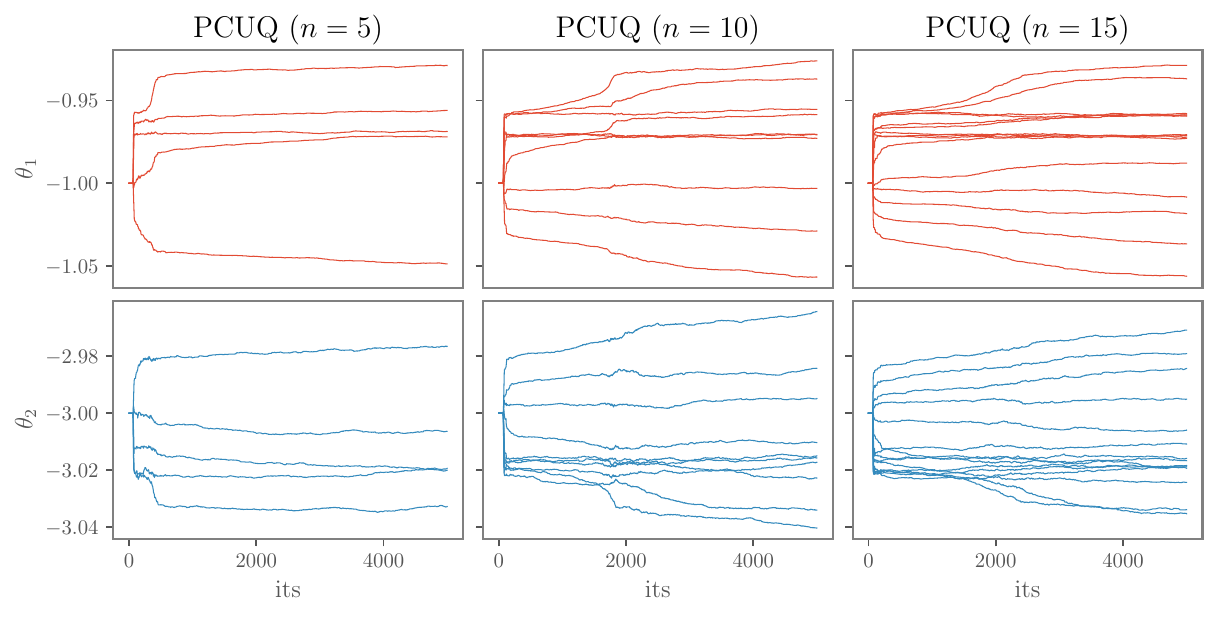}
    \caption{Lotka--Volterra Model (modified parameter settings).  
    Investigating sensitivity to the number of particles $N$.  
    The proposed gradient flow algorithm involves selecting a number of particles $N$, and in the main text we presented results based on $N = 10$. 
    To assess the sensitivity to this choice, we re-computed trace plots based on the alternative choices $N=5$ (left) and $N=15$ (right).
    Here data were generated from the misspecified model. }
    \label{fig: lv5 mod}
\end{figure}

\subsection{Cell Signalling}
\label{app: cell}

\paragraph{ERK Signalling Model}

The \ac{erk} pathway is a chain of proteins in mammalian cells that communicates a signal from a receptor on the surface of the cell to the DNA in the nucleus of the cell (see the top panel of \Cref{fig: cell model}).
In many cancers (e.g. breast, melanoma), a defect in the \ac{erk} pathway leads to uncontrolled growth, and as such this pathway has been the subject of considerable research effort.
This has included the development of detailed mathematical models of \ac{erk} signalling, and the development of compounds that can inhibit steps in the \ac{erk} pathway, providing one route to treatment \citep{hilger2002ras}.

For the purposes of this case study, we considered one such mathematical model for \ac{erk} signalling based on a system of coupled \acp{ode}, due to \citet{kwang2003mathematical}.
This model \citep{erk}, which considers just a subset of the species depicted in \Cref{fig: cell model}, was downloaded from the \texttt{BioModels} repository \citep{BioModels2020} on 4 September 2024.
It is well-known that protein signalling in mammalian cells is intrinsically stochastic due to finite molecular counts and the spatially heterogeneous, compartmentalised, dynamic cellular environment.
This is in contrast to the high copy numbers and spatially homogeneous (``well mixed'') environment that are implicitly assumed in an \ac{ode} model \citep{oates2012network}.
Yet, deterministic models are prevalent in Systems Biology, in part due to the valuable use of \acp{ode} as a communicative and logico-deductive tool.
There is therefore a real need for statistical techniques that can exploit the rich expert knowledge encoded in an \ac{ode} model whilst acknowledging that such an \ac{ode} model does not necessarily meet the criteria that we would expect from a statistical model; i.e. an \ac{ode} + noise model it is likely to be substantially misspecified.

The model of \ac{erk} signalling that we consider takes the form of a coupled system of 11 differential equations, which we write as $\mathrm{d}u / \mathrm{d}x = f(u;\beta)$, representing the evolution of concentrations of the 11 molecular species shown in \Cref{tab: proteins}.
The full system of \acp{ode} is displayed in \Cref{fig: ode erk}, and contains a total of  11 non-negative rate parameters $\beta = (\beta_1 , \dots , \beta_{11})$ that collectively determine the signalling dynamics, together with the initial conditions in the right hand column of \Cref{tab: proteins}.
Here we consider noisy observations $y_{1:n}$ of the concentrations $u$ at times $x_{1:n}$, where the noise follows an additive Gaussian measurement error model \eqref{eq: Gauss likelihood}.
This case study is intended to represent the situation where patient- or cell line-specific kinetic parameters are to be inferred from time-course proteomic data; a fundamentally important scientific task \citep[e.g. as considered by the authors][who developed this model]{kwang2003mathematical}.

\begin{table}[h!]
\centering
\small
\begin{tabular}{|c|c|c|} \hline
    Variable Name & Protein/Protein Complex & Initial Concentration \\ \hline
    $u_1$ & [Raf1] & 2 \\
    $u_2$ & [RKIP] & 2.5 \\
    $u_3$ & [Raf1\_RKIP] & 0 \\
    $u_4$ & [Raf1\_RKIP\_ERKPP] & 0 \\
    $u_5$ & [ERK] & 0 \\
    $u_6$ & [RKIPP] & 0 \\
    $u_7$ & [MEKPP] & 2.5 \\
    $u_8$ & [MEKPP\_ERK] & 0 \\
    $u_9$ & [ERKPP] & 2.5 \\
    $u_{10}$ & [RP] & 3 \\
    $u_{11}$ & [RKIPP\_RP] & 0 \\ \hline
\end{tabular}
\caption{Protein species in the \ac{erk} signalling model of \citep{erk}. }
\label{tab: proteins}
\end{table}

\begin{figure}[h!]
\small
\begin{align*}
    \frac{\mathrm{d}u_1}{\mathrm{d}x} & = f_1(u; \beta) = - \beta_1 u_1 u_2 + \beta_2 u_3 + \beta_5 u_4 \\ 
    \frac{\mathrm{d}u_2}{\mathrm{d}x} &= f_2(u; \beta) = - \beta_1 u_1 u_2 + \beta_2 u_3 + \beta_{11} u_{11} \\
    \frac{\mathrm{d}u_3}{\mathrm{d}x} &= f_3(u; \beta) = \beta_1 u_1 u_2 - \beta_2 u_3 - \beta_3 u_3 u_9 + \beta_4 u_4 \\
    \frac{\mathrm{d}u_4}{\mathrm{d}x} &= f_4(u; \beta) =\beta_3 u_3 u_9 - \beta_4 u_4 - \beta_5 u_4 \\
    \frac{\mathrm{d}u_5}{\mathrm{d}x} &= f_5(u; \beta) =\beta_5 u_4 - \beta_6 u_5 u_7 + \beta_7 u_8 \\
    \frac{\mathrm{d}u_6}{\mathrm{d}x} &= f_6(u; \beta) =\beta_5 u_4 - \beta_9 u_6 u_{10} + \beta_{10} u_{11} \\
    \frac{\mathrm{d}u_7}{\mathrm{d}x} &= f_7(u; \beta) =- \beta_6 u_5 u_7 + \beta_7 u_8 + \beta_8 u_8 \\
    \frac{\mathrm{d}u_8}{\mathrm{d}x} &= f_8(u; \beta) =\beta_6 u_5 u_7 - \beta_7 u_8 - \beta_8 u_8 \\
    \frac{\mathrm{d}u_9}{\mathrm{d}x} &= f_9(u; \beta) =- \beta_3 u_3 u_9 + \beta_4 u_4 + \beta_8 u_8 \\
    \frac{\mathrm{d}u_{10}}{\mathrm{d}x} &= f_{10}(u; \beta) =- \beta_9 u_6 u_{10} + \beta_{10} u_{11} + \beta_{11} u_{11} \\
    \frac{\mathrm{d}u_{11}}{\mathrm{d}x} &= f_{11}(u; \beta) = \beta_9 u_6 u_{10} - \beta_{10} u_{11} - \beta_{11} u_{11} 
\end{align*}
\caption{\ac{ode} model of \ac{erk} signalling due to \citet{erk}.  The variables $u_i$ represent protein concentrations, as defined in \Cref{tab: proteins}, while $\beta = (\beta_1 , \dots , \beta_{11})$ are non-negative rate parameters to be inferred. }
\label{fig: ode erk}
\end{figure}

\paragraph{Parameter Settings}

The data-generating parameters that we used for the \ac{erk} model were:
$\beta_1 = 0.53$, $\beta_2 = 0.0072$, $\beta_3 = 0.625$, $\beta_4 = 0.00245$, $\beta_5 = 0.0315$, $\beta_6 = 0.8$, $\beta_7 = 0.0075$, $\beta_8 = 0.071$, $\beta_9 = 0.92$, $\beta_{10} = 0.00122$, $\beta_{11} = 0.87$, taken from \citet{erk}.
Similarly to \Cref{subsec: ecology} we work with transformed parameters $\theta_i = \text{logit}(\beta_i)$ so that the sign of the $\theta_i$ is unconstrained.
The distribution $Q_0$ (i.e. the prior in the standard Bayesian context) was again taken to be standard multivariate normal.

\paragraph{Data Generation}

The true molecular concentrations $u(x)$ were simulated from the same model, but where the propensity for molecular interaction was no longer constant but slightly time-varying, reflecting the complexity of the cellular environment in which signalling occurs \citep{oates2012network}.
Specifically, we simulated data based on the \ac{ode} model
\begin{align*}
    \frac{\mathrm{d}u}{\mathrm{d}x} = f(u; \beta(x))
\end{align*}
with $\beta_i(x)=\left(1+1/2\sin(2\pi x/45 ) \right)\beta_i$ simulated from $x=0$ to $x=60$.
Data were extracted at times $x_{1:80}$ ranging from $0$ to $45$ in increments of $\mathrm{d}x = 0.57$, with Gaussian measurement noise of variance $\sigma^2_i = 0.01\hat{\sigma}^2_i$ added, where $\hat{\sigma}_i$ is a measure of the characteristic scale for the $i$th molecular concentration, defined as the standard deviation of $u_i(X)$ for a time $X$ sampled from $\mathrm{Uniform}(0,60)$.

\paragraph{Causal Prediction Task}

One targeted treatment for cancer sub-types characterised by dysregulated \ac{erk} signalling is a \ac{mek} inhibitor compound.
A \acf{meki} is a small molecule which is absorbed by the cancer cell and binds to the \ac{mek} protein, negatively affecting its kinase functionality.
As such, \acp{meki} have been considered as a potential cancer treatment (e.g. of melonoma) \citep{wang2007clinical}, and examples of \acp{meki} include the drugs \emph{Trametinib}, \emph{Binimetinib}, and \emph{Cobimetinib}.
The prediction task that we considered here was to predict the evolution of molecular concentrations -- and in particular \ac{erk} -- under the action of a \ac{meki} treatment. 
This is a causal prediction task, since we only have data from the regime where the \ac{meki} was not used, and we aim to leverage the explicitly causal \ac{ode} model to reason about the effect of the treatment.
Concretely, we replace $\beta_6$ with $\gamma\beta_6$ on the right hand side of \eqref{fig: ode erk} (i.e. an effective reduction in the concentration of the functionally active phosphorylated form of \ac{mek}) to obtain the causal model corresponding to the drug-treated cell.
Here we fixed $\gamma = 0.01$ %
(as opposed to setting $\gamma$ exactly to 0) to acknowledge the limited efficacy of a \ac{mek} inhibitor treatment.
A key goal is to predict the concentration of \ac{erk} (i.e. $u_5$) and related protein complexes following treatment; this is a natural functional endpoint, since activated \ac{erk} translocates to the cell nucleus and regulate gene expression via the phosphorylation of transcription factors \citep{yeung2000mechanism}.

Full results from standard Bayesian inference are shown in \Cref{fig: full Bayes erk}.
Full results from \ac{pcuq} are shown in \Cref{fig: full pcuq erk}.
All setting for \ac{pcuq} were identical to those used in \Cref{subsec: ecology}.

\begin{figure}
    \includegraphics[width = \textwidth,clip,trim = 0cm 0cm 25.5cm 0cm]{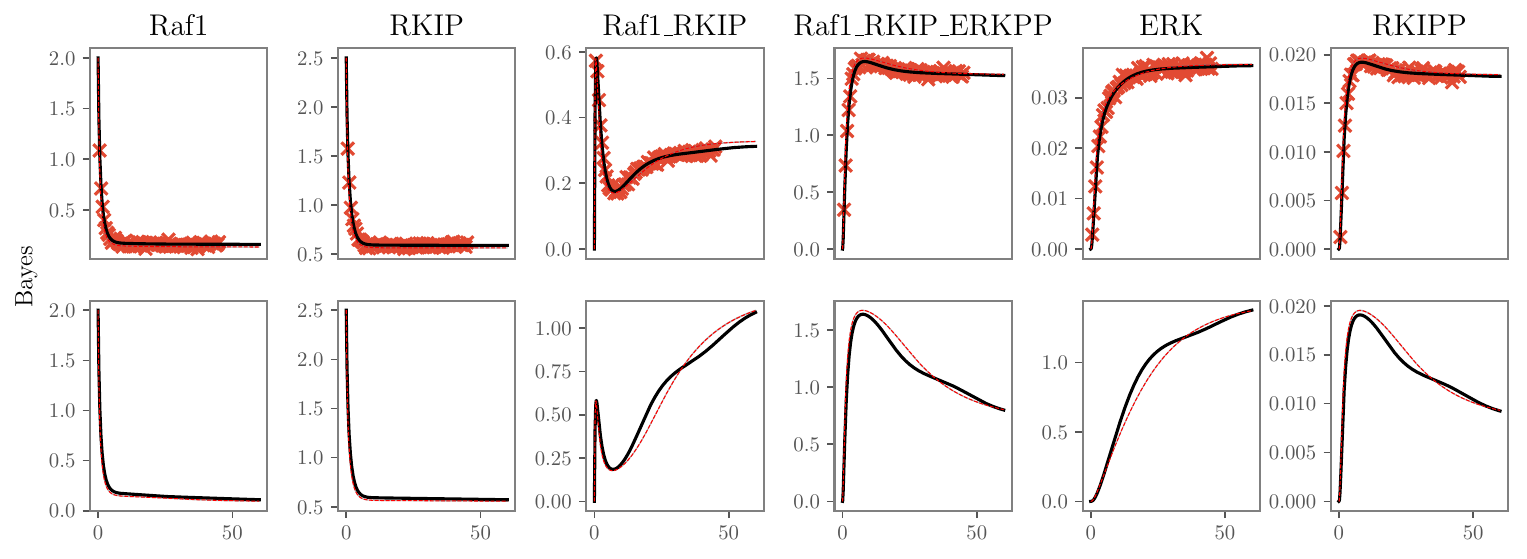}
    
    \hspace{14pt} \includegraphics[width = \textwidth,clip,trim = 26cm 0cm 0cm 0cm]{figs/Bayes_erk.pdf}
    \caption{Full results for the ERK signalling model.  Data (red crosses) were obtained as noisy observations of the system trajectory (black line) in an observational regime (top row), and the task was to predict the causal effect of a MEK inhibitor (MEKi) on ERK signalling (bottom row).
    Predictions produced using standard Bayesian inference were grossly over-confident.
    [Here the red line indicates predictive mean, while the shaded region (so small as to be difficult to see) indicates predictive quartiles.]}
    \label{fig: full Bayes erk}
\end{figure}

\begin{figure}
    \includegraphics[width = \textwidth,clip,trim = 0cm 0cm 25.5cm 0cm]{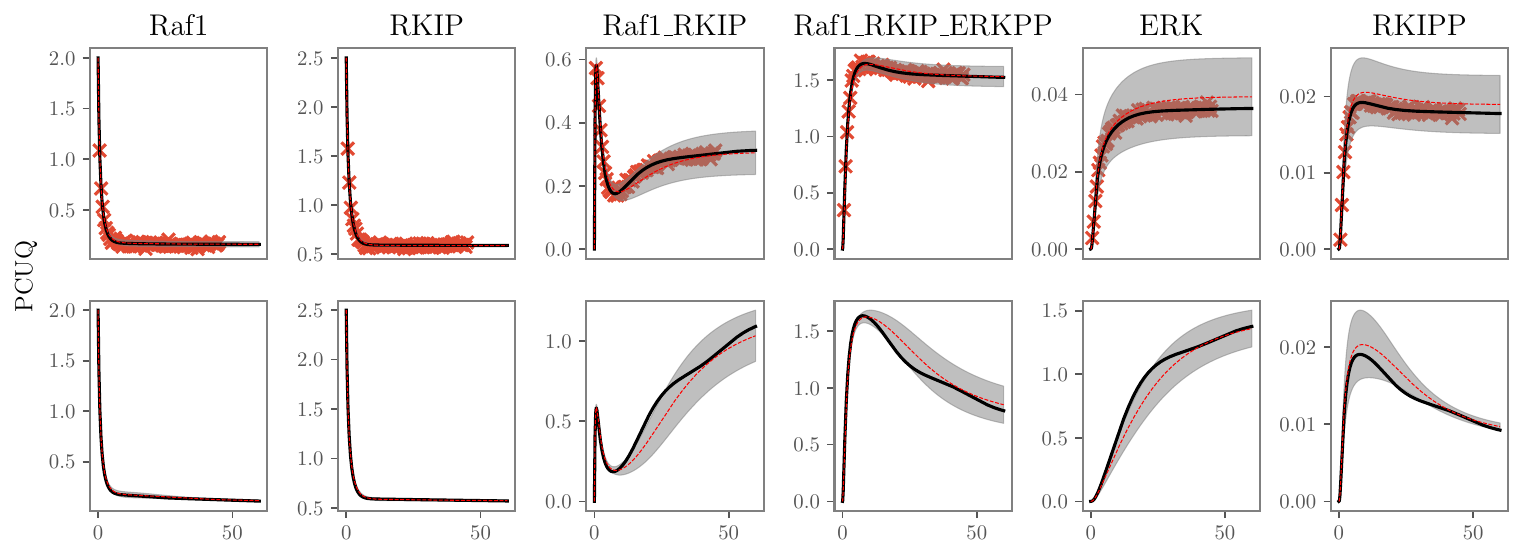}
    
    \hspace{14pt} \includegraphics[width = \textwidth,clip,trim = 26cm 0cm 0cm 0cm]{figs/pcuq_erk.pdf}
    \caption{Full results for the ERK signalling model.  Data (red crosses) were obtained as noisy observations of the system trajectory (black line) in an observational regime (top row), and the task was to predict the causal effect of a MEK inhibitor (MEKi) on ERK signalling (bottom row).
    Predictions produced using \ac{pcuq} enabled predictive uncertainty to be accurately quantified.
    [Here the red line indicates predictive mean, while the shaded region indicates predictive quartiles.]}
    \label{fig: full pcuq erk}
\end{figure}